\documentclass[12pt]{article}
\usepackage[utf8]{inputenc}

\usepackage{geometry,amsmath,mathtools,blkarray,amssymb,amsthm,setspace,xcolor,bbm,tikz-cd,xparse,amsfonts,authblk,graphics,datetime,thmtools}

\usepackage[normalem]{ulem}

\usepackage{caption,subcaption}

\usepackage{tikz,tikz-3dplot}    
\usetikzlibrary{decorations.pathreplacing}

\usepackage[shortlabels]{enumitem}
\usepackage[colorlinks,linkcolor=blue,citecolor=blue,urlcolor=blue,bookmarks=false,hypertexnames=true]{hyperref}
\usepackage[backend=biber, style=bwl-FU, sortcites=false, maxcitenames=2, mincitenames=1, maxbibnames=8, uniquelist=false]{biblatex}

\allowdisplaybreaks
\bibliography{ref.bib}

\theoremstyle{plain}
\newtheorem{theorem}{Theorem}[section]
\newtheorem{proposition}{Proposition}[section]
\newtheorem{lemma}{Lemma}[section]

\theoremstyle{definition}

\newtheorem{example}{Example}
\newtheorem{assumption}{Assumption}[section]
\newtheorem{condition}{Condition}[section]

\newtheorem{remark}{Remark}[section]

\renewcommand\thmcontinues[1]{Continued}

\newtheorem{conditionalt}{Condition}[condition]
\newenvironment{conditionp}[1]{
  \renewcommand\theconditionalt{#1}
  \conditionalt
}{\endconditionalt}

\newcommand{\indep}{\perp \!\!\! \perp}
\newcommand{\norm}[1]{\left\|#1\right\|}

\newcommand{\at}[1]{\ \bigg |_{#1}}

\NewDocumentCommand{\defmathletter}{m}{%
    \expandafter\newcommand\csname b#1\endcsname{\mathbb{#1}}%
    \expandafter\newcommand\csname c#1\endcsname{\mathcal{#1}}%
}
\NewDocumentCommand{\defmathletters}{>{\SplitList{,}}m}{\ProcessList{#1}{\defmathletter}}
\defmathletters{A,B,C,D,E,F,G,H,I,J,K,L,M,N,O,P,Q,R,S,T,U,V,X,Y,Z}

\NewDocumentCommand{\defvector}{m}{%
    \expandafter\newcommand\csname v#1\endcsname{\mathbf{#1}}%
}
\NewDocumentCommand{\defvectors}{>{\SplitList{,}}m}{\ProcessList{#1}{\defvector}}
\defvectors{A,B,C,D,E,F,G,H,I,J,K,L,M,N,O,P,Q,R,S,T,U,V,X,Y,Z}
\defvectors{a,b,c,d,e,f,g,h,i,j,k,l,m,n,o,p,q,r,s,t,u,v,x,y,z}

\newdateformat{monthyeardate}{\monthname[\THEMONTH] \THEYEAR}

\title{Semiparametric Efficiency Gains From Parametric Restrictions on Propensity Scores}
\author{Haruki Kono \thanks{Email: hkono@mit.edu. 
I am grateful to Alberto Abadie and Whitney Newey for advising this project; to the associate editor and two anonymous referees for insightful review comments; to Vod Vilfort for careful proofreading; to Amy Finkelstein, Stephen Morris, and Victor Orestes for their guidance in the second year paper course; and to, Yifan Dai, Sivakorn Sanguanmoo and participants in MIT econometrics lunch seminar for their feedback and helpful discussions.}}
\affil{MIT}
\date{\monthyeardate\today}

\begin{document}

\maketitle

\begin{abstract}
We explore how much knowing a parametric restriction on propensity scores improves semiparametric efficiency bounds in the potential outcome framework.
For stratified propensity scores, considered as a parametric model, we derive explicit formulas for the efficiency gain from knowing how the covariate space is split.
Based on these, we find that the efficiency gain decreases as the partition of the stratification becomes finer.
For general parametric models, where it is hard to obtain explicit representations of efficiency bounds, we propose a novel framework that enables us to see whether knowing a parametric model is valuable in terms of efficiency even when it is high-dimensional.
In addition to the intuitive fact that knowing the parametric model does not help much if it is sufficiently flexible, we discover that the efficiency gain can be nearly zero even though the parametric assumption significantly restricts the space of possible propensity scores.
\end{abstract}

\section{Introduction}

Let $T$ be a treatment indicator such that $T = 1$ if an individual is treated, and $Y_1$ and $Y_0$ denote potential outcomes for $T = 1$ and $T = 0,$ respectively.
For each individual, we observe the treatment status $T,$ the realized outcome $Y = T Y_1 + (1 - T) Y_0,$ and a covariate vector $X,$ which takes values on some set $\cX.$
We consider the average treatment effect on the treated, $\mathrm{ATT} = \bE [Y_1 - Y_0 \mid T = 1].$
Suppose that the propensity score $p^\ast (x) = \bP(T = 1 \mid X = x)$ belongs to a finite-dimensional parametric family.
As usual, we assume the conditional independence and the overlap condition, which are formally defined in Assumption \ref{ass:unconf+overlap}.
We consider three different situations: (K) $p^\ast$ is known, (P) $p^\ast$ is known to belong to the parametric family, but the true parameter is unknown, and (UK) $p^\ast$ is completely unknown.
Let $V^k,$ $V^p,$ and $V^{uk}$ be the semiparametric efficiency bounds of the ATT under (K), (P), and (UK), respectively.
Clearly, the model under (P) is smaller than the model under (UK) because $p^\ast$ is parametrically specified.
On the other hand, the model under (P) is larger than the model under (K) because the true parameter is not available.
These observations imply $V^k \leq V^p \leq V^{uk}.$

This kind of comparison among the efficiency bounds for different models has received attention in causal inference.
\cite{hahn1998role} derives the formulas for $V^k$ and $V^{uk}$ and shows the strict inequality $V^k < V^{uk}$ holds generically, which is generalized to multivalued treatment environments by \cite{lee2018efficient}.
Also, \cite{chen2008semiparametric} derive the formula for $V^p$ in the missing data literature.
In this paper, we provide a quantitative and qualitative analysis of the impact of various parametric restrictions on the propensity score on the estimation efficiency of causal parameters defined via general moment conditions, allowing for multivalued treatments.

One of the key contributions of this paper is to obtain explicit formulas of the main components of the efficiency gains $V^{uk} - V^p$ and $V^p - V^k$ when the propensity score is known to be \textit{stratified}, that is, there is a known finite partition of the covariate space $\cX$ such that the propensity score is constant on each region.
Notice that this class of propensity scores can be seen as a parametric family parametrized by the treatment assignment probabilities on each region.
The stratified propensity score appears in many empirical studies as reviewed below, and importantly, it covers cases where the propensity score depends only on categorical covariates.
The formulas of the efficiency gains provide insightful implications.
First, once we know how the covariate space is split, we can match observations that lie in the same region of the partition, which enables the efficient estimation of the parameters of interest conditional on each region.
This leads to the efficiency improvement $V^{uk} - V^p,$ as it removes heterogeneity within each subclass.
If we know the values of the propensity score in each region in addition to how the space is split, we can efficiently aggregate the local estimates we obtain, which makes the other improvement $V^p - V^k.$
When the partition is coarse, i.e., the model is small, the gain from matching observations within the same region is dominant, and consequently, $V^p$ is close to $V^k.$
As the partition becomes finer, i.e., the model becomes larger, the matching effect decreases, so that knowing the partition becomes less beneficial, and $V^p$ and $V^{uk}$ become similar.

The idea that as the parametric model of the propensity score gets larger, the bound $V^p$ also gets larger should hold beyond the stratified propensity score.
Unfortunately, it is hard to obtain explicit formulas of efficiency gains for other parametric models, because $V^p$ involves a complicated matrix inversion, which is solvable only when the propensity score is stratified.
Instead of relying on such formulas, we introduce a new framework.
We first define a growing sequence of parametric models of the propensity score, keeping the ``structure'' of the model fixed.
We then consider efficiency bounds along the sequence.
We are interested in whether the limit of the sequence of efficiency bounds $V^p$ is equal to the nonparametric bound $V^{uk}.$
If it is, it implies that knowing the parametric structure is nearly equivalent to knowing nothing in terms of efficiency, especially when the model is high-dimensional.
Otherwise, it means that the parametric structure intrinsically narrows down the space of possible propensity scores.
Theorem \ref{thm:many-param-limit} shows that $V^p$ approaches $V^{uk}$ if and only if the set of scores of the parametric model approximates two specific functions that are composed of the propensity scores and the mean regression functions of moment functions.
This result implies the intuitive fact that if a parametric model is sufficiently flexible, then assuming it is almost equivalent to assuming a nonparametric model.
More importantly, it also tells us that as long as the condition of Theorem \ref{thm:many-param-limit} is satisfied, it may hold that $V^{uk} \approx V^p,$ even though the parametric model is not much flexible and does not approximate all functions in the nonparametric model.
In other words, there are restrictive parametric models under which propensity scores are almost ancillary for estimating parameters of interest.

\medskip

\noindent
\textit{Related Literature.}
In the binary treatment setup, the propensity score method is proposed by \cite{rosenbaum1983central} and \cite{rosenbaum1984reducing}.
\cite{hahn1998role} computes asymptotic variance bounds for the average treatment effect and the ATT and proposes efficient estimators of them by applying the theory of semiparametric estimation developed by \cite{newey1990semiparametric} and \cite{bickel1993efficient}.
\cite{hirano2003efficient} give another efficient estimator called the inverse probability weighted estimator, of which finite sample performances are investigated in \cite{herren2023true}.
For the quantile treatment effect, which is often of interest for measuring inequality, \cite{firpo2007efficient} obtains similar efficiency results.
In the missing data literature, \cite{chen2008semiparametric} derive the semiparametric efficiency bound for the parameter conditional on the treated subpopulation when the propensity score is correctly specified by a parametric model.

Theoretical studies on the estimation of treatment effects in a multivalued treatment setup are initiated by \cite{imbens2000role}, who generalizes the framework of Rosenbaum and Rubin.
For efficient estimation, \cite{cattaneo2010efficient} provides the efficient influence function and the semiparametric efficiency bound for parameters defined via general moment conditions, extending results given by \cite{hahn1998role} and \cite{hirano2003efficient}.
While \cite{cattaneo2010efficient} focuses on multivalued treatment effects on the whole population, \cite{farrell2015robust} considers the efficient estimation of the average treatment effect on the treated, and \cite{lee2018efficient} investigates efficiency bounds for multivalued treatment effects on subpopulation in general, but those for parametric propensity scores are not covered.

In a large fraction of empirical studies, the propensity score is modeled by a parametric model.
In observational studies, where the propensity score is not available to analysts, the propensity score is often specified by the logit or probit model.
See, for example, \cite{yang2016propensity} and \cite{imai2004causal}.
Even in experimental studies, an experiment designer often stratifies participants for the sake of estimation efficiency as in, for example, \cite{chong2016iron}, \cite{bugni2019inference} and \cite{hong2020inference}.
Propensity scores in stratified experiments can be thought of as being parameterized by assigning probabilities on each stratum.

\cite{hahn1998role} points out that knowing the propensity score affects the semiparametric efficiency bound of the ATT, while it is ancillary for the estimation of the average treatment effect.
\cite{frolich2004note} explains why this happens qualitatively, which is valid even in multivalued treatment cases.
The argument by \cite{chen2008semiparametric}, who examine the semiparametric efficiency bound when the propensity score is correctly specified by a parametric model in the binary setup, reveals that knowing the parametric model where the true propensity score lives improves the efficiency, but the amount of the improvement is unclear because the efficiency bound has an analytically complicated representation.
In Section \ref{sec:stratified-experiment}, we avoid this issue by focusing on the stratified experiment setup, which is analytically tractable.
It is also noteworthy that our results also include a direct answer to the conjecture raised in page 325 of \cite{hahn1998role}: does the knowledge about the propensity score matter in stratified experiments?

The experimental design literature, including \cite{cytrynbaum2021designing}, \cite{tabord2023stratification}, and \cite{bai2024efficiency} among others, has explored the optimal way to stratify the covariate space in order to achieve the efficient estimation by design.
In contrast, in Section \ref{sec:stratified-experiment}, we derive the efficiency bound for a fixed stratification and examine how the bound changes as the stratification becomes finer.

To the best of our knowledge, \cite{lee2018efficient} is the only study that investigates the value of the knowledge of the propensity score in the multivalued treatment setup.
The paper does so by considering how propensity scores are incorporated into the efficient influence function.
Although \cite{lee2018efficient} is the closest work to this paper, the paper models the ``partial knowledge'' of the propensity score differently than we do.
\cite{lee2018efficient} considers the cases where an analyst knows the propensity scores for some treatments but not for the others.
On the other hand, we think of the ``partial knowledge'' as in which parametric model the propensity score lives.
In this regard, \cite{lee2018efficient} and the present paper complement each other.

The theory we propose in Section \ref{sec:param-vs-nonparam} may look similar to the notion of sieves in that both consider a growing sequence of parametric models; see \cite{chen2007large} for a thorough review.
However, they have completely distinct goals.
Sieves are typically used when the target parameter is defined as the solution of an infinite-dimensional optimization problem, which is difficult to compute in general.
The sieve method considers a growing sequence of lower dimensional spaces that approximates the original space.
By solving the problem constrained on the approximating space, it provides flexible and robust estimators of the infinite-dimensional parameter.
On the other hand, this paper addresses the following question: if the true parameter lives in a low-dimensional parameter space, but an analyst assumes a high-dimensional model, how much efficiency would he or she lose?
Thus, we are interested in any increasing sequence of parametric models, including one that does not approximate the nonparametric model.

\section{Setup and Identification} \label{sec:setup}

Suppose that there are a finite number of treatments $\cT = \{0, 1, \cdots, J\}$ and corresponding potential outcomes $Y_0, Y_1, \cdots, Y_J \in \bR,$ assuming no interference.
Let $T$ be a $\cT$-valued random variable that indicates which treatment is assigned.
The observed outcome is, therefore, $Y \coloneqq \sum_{j \in \cT} D_j Y_j$ where $D_j \coloneqq \bI\{T = j\}$ and $\bI \{\cdot\}$ is the indicator function.
In addition, a random covariate vector $X,$ which is fully supported on a set $\cX,$ is also available to the analyst.
The dataset consists of $W \coloneqq (Y, X, T)$ of all participants.
Let $F_X$ and $F_{D, X}$ be the distributions of $X$ and $(D, X),$ respectively, which are not available to the analyst.

We describe the parameter of interest following \cite{cattaneo2010efficient} and \cite{lee2018efficient}.
Let $\cS \subset \cT$ be a nonempty set of treatment types.
The target parameter is $\beta_{\cS}^{\ast} = ((\beta_{0 \mid \cS}^\ast)^\prime, \cdots, (\beta_{J \mid \cS}^\ast)^\prime)^\prime \in \bR^{(J + 1) d_\beta}$ where $\beta_{j \mid \cS}^\ast \in \bR^{d_\beta}$ is defined as a solution of
\begin{align*}
    0
    =
    \bE[m(Y_j; \beta_{j \mid \cS}) \mid T \in \cS]
\end{align*}
for a known possibly non-smooth moment function $m : \bR \times \bR^{d_\beta} \to \bR^{d_m}$ where $d_\beta \leq d_m$ allowing for over-identification.
In what follows, we omit the subscript $\cS$ of parameters that indicates the conditioning set if it does not cause a confusion, so that $\beta_{\cS}$ and $\beta_{j \mid \cS}$ will be just written as $\beta$ and $\beta_j,$ respectively.

To identify the parameter, we impose the following assumption.
\begin{assumption} \label{ass:id}
For each $j \in \cT,$ $\beta_j = \beta_j^\ast$ is the unique solution of $\bE[m(Y_j; \beta_j) \mid T \in \cS] = 0.$
\end{assumption}

This framework unifies many interesting causal parameters.
In the binary treatment case $\cT = \{0, 1\},$ for example, the average treatment effect corresponds to $m(y_j, \beta_j) = y_j - \beta_j$ and $\cS = \cT.$
Similarly, one can deal with the ATT by considering $\cS = \{1\}.$
For $\tau \in (0, 1),$ the moment function $m(y_j, \beta_j) = \bI\{y_j \leq \beta_j\} - \tau$ leads to the quantile treatment effect at $\tau$th quantile.
Also, suppose that there are three treatments: $\cT = \{0, 1, 2\},$ and that a policy maker is considering abolishing treatments $0$ and $1$ and transferring people who are assigned to them to treatment $2.$
In this case, she would be interested in $\bE[Y_2 \mid T \in \{0, 1\}],$ the average potential outcome of treatment $2$ for those with $0$ and $1.$

The propensity score $p_j^\ast(x) \coloneqq \bP(T = j \mid X = x)$ is assumed to be correctly specified by a $d_\gamma$-dimensional smooth parametric model $\gamma \mapsto p_j(\cdot; \gamma).$
That is, $p_j^\ast (\cdot) = p_j(\cdot; \gamma^\ast)$ for some $\gamma^\ast.$
In abuse of notation, we also denote $p_j(\gamma) \coloneqq \bE[p_j(X; \gamma)]$ and $p_j^\ast \coloneqq p_j(\gamma^\ast).$
For $\cS \subset \cT,$ let $p_{\cS} \coloneqq \sum_{j \in \cS} p_j$ and $D_{\cS} \coloneqq \sum_{j \in \cS} D_j.$

Following the literature, we assume the ignorability.
\begin{assumption} \label{ass:unconf+overlap}
\mbox{}
\begin{description}
    \item[(Unconfoundedness)] $T \indep Y_j \mid X$ for all $j \in \cT,$
    \item[(Overlap)] There exists $p_{\mathrm{min}} > 0$ such that $p_j^\ast (X) > p_{\mathrm{min}}$ almost surely for all $j \in \cT.$
\end{description}
\end{assumption}

The first condition ensures that one can compare outcomes with different treatments conditioning on covariates.
The second condition is sufficient for identifying the parameter and for the semiparametric efficiency bound to be finite.
Under Assumption \ref{ass:unconf+overlap}, the parameter of interest is identified as
\begin{align*}
    0
    =
    \bE \left[e_j^\ast (X; \beta_j) \frac{p_{\cS}^{\ast} (X)}{p_{\cS}^{\ast}}\right]
    \Leftrightarrow
    \beta_j = \beta_j^\ast
\end{align*}
where $e_j^\ast (X; \beta_j) \coloneqq \bE [m(Y; \beta_j) \mid X, T = j].$
Note that since $p_{\cS}^{\ast}$ is constant, it is irrelevant to the identification.

The score function $S_j$ for the parametric model of the propensity score is defined as
\begin{align*}
    S_j (x; \gamma)
    \coloneqq
    \frac{\partial}{\partial \gamma}
    \log p_j (x; \gamma)
    .
\end{align*}
Similarly, let $S_{\cS} (x; \gamma) \coloneqq \frac{\partial}{\partial \gamma} \log p_{\cS, \gamma} (x).$
Note that $\sum_{j \in \cS} S_j \neq S_{\cS}$ in general.
Let $S_j^\ast (x) \coloneqq S_j (x; \gamma^\ast).$
For the regularity of this parametric propensity score, we assume the following condition.
For a column vector $a,$ define $a^{\otimes 2} \coloneqq a a^\prime.$
\begin{assumption} \label{ass:full-rank-parametric}
$\bE \left[\sum_{j \in \cT} D_j S_j^\ast (X)^{\otimes 2} \right]$ exists and is invertible.
\end{assumption}
This condition essentially requires that the Fisher information of the parametric model is invertible, which guarantees that the parameterization $\gamma \mapsto p_j(\cdot; \gamma)$ is not degenerate, i.e., no components of $\gamma$ are redundant.
A similar condition is assumed in Theorem 3 of \cite{chen2008semiparametric}.

As in \cite{cattaneo2010efficient}, we impose the following assumption.

\begin{assumption} \label{ass:full-rank-gradient}
For each $j \in \cT,$ it holds that $\bE \left[ \norm{m (Y_j; \beta_j^\ast)}^2 \right] < \infty,$ where the Euclidean norm is denoted by $\norm{\cdot},$ and
\begin{align*}
    \cJ_j
    \coloneqq
    \frac{\partial}{\partial \beta_j}
    \at{\beta_j = \beta_j^\ast}
    \bE\left[m(Y_j; \beta_j)\mid T \in\cS\right]
    \in
    \bR^{d_m \times d_\beta}
\end{align*}
is column full-rank.
\end{assumption}
This is another condition for the finiteness of the efficiency bound.
Note that we implicitly assume that $\bE\left[m(Y_j; \beta_j)\mid T \in\cS\right]$ is differentiable in $\beta_j,$ although the moment function $m$ itself may be non-smooth.

\section{Semiparametric Efficiency Bounds} \label{sec:speb}

\subsection{Derivation}
In this section, we derive the efficient influence function and semiparametric efficiency bound for $\beta$ in the framework offered by \cite{bickel1993efficient} and \cite{hahn1998role}.

Fix $\cS \subset \cT.$
Define functions $s_j,$ $c_j,$ and $t_j$ as
\begin{gather*}
    s_j(W; \beta_j, e_j, \gamma)
    \coloneqq
    \frac{p_{\cS}(X; \gamma)}{p_{\cS} (\gamma) p_j(X; \gamma)}
    (m(Y; \beta_j) - e_j(X; \beta_j))
    ,
    \\
    c_j(\beta_j, e_j, \gamma)
    \coloneqq
    \frac{1}{p_{\cS} (\gamma)}
    \bE \left[
        e_j(X; \beta_j)
        \sum_{i \in \cS}
        D_i
        S_i (X; \gamma)^\prime
    \right]
    \bE \left[
        \left(
            \sum_{i \in \cT} D_i S_i(X; \gamma)
        \right)^{\otimes 2}
    \right]^{-1}
    ,
    \\
    t_j(W; \beta_j, e_j, \gamma)
    \coloneqq
    \frac{p_{\cS}(X; \gamma)}{p_{\cS} (\gamma)}
    e_j(X; \beta_j)
    .
\end{gather*}
With these functions in hand, let $F^p \coloneqq ((F_0^p)^\prime, \dots, (F_J^p)^\prime)^\prime,$ where
\begin{align*}
    F_j^p (W)
    \coloneqq
    D_j
    s_j(W; \beta_j^\ast, e_j^\ast, \gamma^\ast)
    +
    c_j(\beta_j^\ast, e_j^\ast, \gamma^\ast) 
    \sum_{i \in \cT} 
    D_i
    S_i^\ast(X)
    +
    t_j(W; \beta_j^\ast, e_j^\ast, \gamma^\ast)
\end{align*}
for $j \in \cT.$

The three functions, $s_j,$ $c_j,$ and $t_j$ are key components of the efficient influence function of $\beta$ that is derived in Theorem \ref{thm:eff-bound} below.
They correspond to the scores of the distribution of $Y$ conditional on $X$ and $T = j,$ the propensity score, and the marginal distribution of $X,$ respectively.

Also, let
\begin{align*}
    \cJ
    \coloneqq
    \mathrm{diag} (\cJ_0, \cdots, \cJ_J)
    \in
    \bR^{(J + 1) d_m \times (J + 1) d_\beta}
    ,
\end{align*}
which is the matrix that has $\cJ_0, \cdots, \cJ_J$ on its block diagonal.
Note that $\cJ$ is column full-rank under Assumption \ref{ass:full-rank-gradient}.
The following theorem gives the efficient influence function and efficiency bound of $\beta.$

\begin{theorem} \label{thm:eff-bound}
Suppose that Assumptions \ref{ass:id}-\ref{ass:full-rank-gradient} hold.
The efficient influence function of $\beta$ is
\begin{align*}
    \psi^p (W) 
    \coloneqq
    - 
    \left(\cJ^\prime \bE \left[F^p (W)^{\otimes 2}\right]^{-1} \cJ\right)^{-1}
    \cJ^\prime 
    \bE \left[F^p (W)^{\otimes 2}\right]^{-1} 
    F^p (W)
    .
\end{align*}
Consequently, the semiparametric efficiency bound is 
\begin{align*}
    V^p
    \coloneqq
    \left(
        \cJ^\prime \bE \left[F^p (W)^{\otimes 2}\right]^{-1} \cJ
    \right)^{-1}
    .
\end{align*}
\end{theorem}

This result is a direct generalization of Theorem 3 of \cite{chen2008semiparametric} to the multivalued treatment environment.
Indeed, when the treatment variable is binary, i.e., $\cJ = \{0, 1\},$ our result coincides with theirs.
Recall that \cite{lee2018efficient} considers the efficiency in the cases where the propensity score is partially known, i.e., $p_j^\ast$ is known for some $j$ but not for the others.
On the other hand, we model the ``partial knowldge'' of the propensity score by imposing a parametric assumption.

Efficiency bounds are usually used to see whether a given estimator is semiparametrically efficient, i.e., no regular estimator has a smaller asymptotic variance than it does.
In addition to this, they can be used to construct efficient estimators.
For example, \cite{cattaneo2010efficient} proposes an efficient estimator of treatment effects based on the moment condition induced by the efficient influence function.
Estimators based on efficient influence functions are often called debiased machine learning estimators, and \cite{chernozhukov2018double} and \cite{chernozhukov2022locally} show that they are efficient in many examples.
These estimators achieve the efficiency by debiasing biases that arise from estimating nuisance parameters, such as propensity scores.
Efficient influence functions are useful for constructing efficient estimators, especially when nuisance parameters are high-dimensional and difficult to estimate.

\subsection{Examples}

\begin{example}
Consider $m(y_j; \beta_j) = y_j - \beta_j.$ 
The propensity score is assumed to be correctly specified by a parametric model $p_j (\cdot; \gamma).$
The efficiency bound of $\beta_j^\ast = \bE [Y_j \mid T \in \cS]$ is
\begin{align*}
    V_j^p
    =
    \frac{1}{(p_{\cS}^{\ast})^2}
    \left\{
        \bE \left[
            \frac{p_{\cS}^{\ast} (X)^2}{p_j^\ast (X)}
            \sigma_j^2 (X)
        \right]
        +
        \bar c_j^\ast
        \bE \left[
            \sum_{i \in \cT} D_i S_i^\ast (X)^{\otimes 2}
        \right]^{-1}
        (\bar c_j^\ast)^\prime
        +
        \bE \left[
            p_{\cS}^{\ast} (X)^2 (\beta_j^\ast (X) - \beta_j^\ast)^2
        \right]
    \right\}
    ,
\end{align*}
where $\beta_j^\ast (X) = \bE [Y_j \mid X],$ $\sigma_j^2 (X) = \bV (Y_j \mid X),$ where $\bV$ denotes the conditional variance operator, and 
\begin{align*}
    \bar c_j^\ast
    =
    \bE \left[
        (\beta_j^\ast (X) - \beta_j^\ast)
        \sum_{i \in \cS} D_i S_i^\ast (X)^\prime
    \right]
    .
\end{align*}
In particular, for the case of $\cT = \{0, 1\}$ and $\cS = \{1\},$ we obtain the efficiency bound for the ATT, $\beta_{1 \mid \{1\}} - \beta_{0 \mid \{1\}},$
\begin{align} \label{eq:efficiency-bound-ATT}
    V_{\mathrm{ATT}}^p
    =
    &
    \frac{1}{(p_1^\ast)^2}
    \Bigg\{
        \bE \left[
            p_1^\ast (X)
            \sigma_1^2 (X)
        \right]
        +
        \bE \left[
            \frac{p_1^\ast (X)^2}{p_0^\ast (X)}
            \sigma_0^2 (X)
        \right]
        \\
        \nonumber
        +
        &
        \left(\bar c_{1}^\ast - \bar c_{0}^\ast\right)
        \bE \left[
            \sum_{j \in \cT} D_j S_j^\ast (X)^{\otimes 2}
        \right]^{-1}
        \left(\bar c_{1}^\ast - \bar c_{0}^\ast\right)^\prime
        \\
        \nonumber
        +
        &
        \bE \left[
            p_1^\ast (X)^2 (\beta_1^\ast (X) - \beta_0^\ast (X) - (\beta_1^\ast - \beta_0^\ast))^2
        \right]
    \Bigg\}
    .
\end{align}
This gives a counterpart of Theorem 1 of \cite{hahn1998role} when there is a parametric restriction on the propensity score.
\end{example}

\begin{example}
For $\tau \in (0, 1),$ consider $m (y_j; \beta_j) =\bI\{y_j \leq \beta_j\} - \tau.$
Then, the moment condition $0 = \bE [m(Y_j; \beta_j^\ast) \mid T \in \cS]$ implies that $\beta_j^\ast$ is the $\tau$-quantile of $Y_j$ conditional on $T \in \cS.$
For a parametric propensity score $p_j (\cdot; \gamma),$ the efficiency bound of $\beta_j^\ast$ is
\begin{align*}
    V_j^p
    =
    &
    \frac{1}{(f_{j \mid T \in \cS} (\beta_j^\ast) p_{\cS}^{\ast})^2}
    \Bigg\{
        \bE \left[
            \frac{p_{\cS}^{\ast} (X)^2}{p_j^\ast (X)}
            \bV (\bI\{Y_j \leq \beta_j^\ast\} \mid X)
        \right]
        \\
        +
        &
        \bar c_j^\ast
        \bE \left[
            \sum_{i \in \cT} D_i S_i^\ast (X)^{\otimes 2}
        \right]^{-1}
        (\bar c_j^\ast)^\prime
        +
        \bE \left[
            p_{\cS}^{\ast} (X)^2
            (\bP (Y_j \leq \beta_j^\ast \mid X) - \tau)^2
        \right]
    \Bigg\}
\end{align*}
where $f_{j \mid T \in \cS}$ is the density of the distribution of $Y_j$ conditional on $T \in \cS$ and
\begin{align*}
    \bar c_j^\ast
    =
    \bE \left[
        (\bP (Y_j \leq \beta_j \mid X) - \tau)
        \sum_{i \in \cS} D_i S_i^\ast (X)^\prime
    \right]
    .
\end{align*}
When $\cT = \{0, 1\}$ and $\cS = \{1\},$ a similar calculation yields the efficiency bound for the quantile treatment effect on the treated, $\beta_{1 \mid \{1\}} - \beta_{0 \mid \{1\}}:$
\begin{align*}
    V_{\textrm{QTT}}^p
    =
    &
    \frac{1}{(p_1^\ast)^2}
    \Bigg\{
        \bE \left[
            p_1^\ast (X)
            \bV \left(\frac{\bI\{Y_1 \leq \beta_1^\ast\}}{f_{1 \mid T \in \cS} (\beta_1^\ast)} \mid X\right)
            +
            \frac{p_1^\ast (X)^2}{p_0^\ast (X)}
            \bV \left(\frac{\bI\{Y_0 \leq \beta_0^\ast\}}{f_{0 \mid T \in \cS} (\beta_0^\ast)} \mid X\right)
        \right]
        \\
        +
        &
        \left(\frac{\bar c_{1}^\ast}{f_{1 \mid T \in \cS} (\beta_1^\ast)} - \frac{\bar c_{0}^\ast}{f_{0 \mid T \in \cS} (\beta_0^\ast)}\right)
        \bE \left[
            \sum_{j \in \cT} D_j S_j^\ast (X)^{\otimes 2}
        \right]^{-1}
        \left(\frac{\bar c_{1}^\ast}{f_{1 \mid T \in \cS} (\beta_1^\ast)} - \frac{\bar c_{0}^\ast}{f_{0 \mid T \in \cS} (\beta_0^\ast)}\right)^\prime
        \\
        +
        &
        \bE \left[
            p_1^\ast (X)^2
            \left(
                \frac{\bP (Y_1 \leq \beta_1^\ast \mid X) - \tau}{f_{1 \mid T \in \cS} (\beta_1^\ast)}
                - 
                \frac{\bP (Y_0 \leq \beta_0^\ast \mid X) - \tau}{f_{0 \mid T \in \cS} (\beta_0^\ast)}
            \right)^2
        \right]
    \Bigg\}
    .
\end{align*}
This is an analog of equation (10) of \cite{firpo2007efficient} when the propensity score is parametrically specified.
\end{example}

\subsection{Basic Comparison} \label{subsec:comparison}

We compare the efficiency bound for a parametric model derived in Theorem \ref{thm:eff-bound} with the bounds for the cases where the propensity score is known or unknown.
In this subsection, we overview preliminary comparison results.

According to Corollary 1 of \cite{lee2018efficient}, the counterparts of $F^p$ when the propensity score is known or unknown are $F^k \coloneqq ((F_0^k)^\prime, \cdots, (F_J^k)^\prime)^\prime$ and $F^{uk} \coloneqq ((F_0^{uk})^\prime, \cdots, (F_J^{uk})^\prime)^\prime$ where
\begin{align} \label{eq:EIF-known}
    F_j^k (W)
    \coloneqq
    D_j 
    s_j(W; \beta_j^\ast, e_j^\ast, \gamma^\ast)
    +
    t_j(W; \beta_j^\ast, e_j^\ast, \gamma^\ast)
    ,
\end{align}
and
\begin{align} \label{eq:EIF-nonparametric}
    F_j^{uk} (W)
    \coloneqq
    D_j 
    s_j(W; \beta_j^\ast, e_j^\ast, \gamma^\ast)
    +
    \frac{1}{p_{\cS}^{\ast}}
    (D_{\cS} - p_{\cS}^{\ast} (X)) e_j^\ast (X; \beta_j^\ast)
    +
    t_j(W; \beta_j^\ast, e_j^\ast, \gamma^\ast)
    .
\end{align}
Therefore, the efficiency bounds are
\begin{align*}
    V^k
    \coloneqq
    \left(
        \cJ^\prime \bE \left[F^k (W)^{\otimes 2}\right]^{-1} \cJ
    \right)^{-1}
    \ \text{and} \ 
    V^{uk}
    \coloneqq
    \left(
        \cJ^\prime \bE \left[F^{uk} (W)^{\otimes 2}\right]^{-1} \cJ
    \right)^{-1}
    ,
\end{align*}
respectively.

In what follows throughout the paper, we assume the just-identification, i.e., $d_m = d_\beta.$
Although we are interested in the relationship among $V^k,$ $V^p$ and $V^{uk},$ we can focus on the comparison among $\bE \left[F^k (W)^{\otimes 2}\right],$ $\bE \left[F^p (W)^{\otimes 2}\right],$ and $\bE \left[F^{uk} (W)^{\otimes 2}\right],$ as for two positive semidefinite matrices $W_1$ and $W_2,$ it holds
\begin{align} \label{eq:derivative-not-matter}
    W_1
    \geq 
    W_2
    \Leftrightarrow
    \left(\cJ^\prime W_1^{-1} \cJ\right)^{-1}
    \geq
    \left(\cJ^\prime W_2^{-1} \cJ\right)^{-1}
    ,
\end{align}
where for two symmetric matrices $A$ and $B,$ we define $A \geq B$ if and only if $A - B$ is positive semidefinite.
Considering the complexity of the models, it is qualitatively obvious that 
\begin{align} \label{eq:comparison-baseline}
    \bE \left[F^k (W)^{\otimes 2}\right] 
    \leq 
    \bE \left[F^p (W)^{\otimes 2}\right] 
    \leq 
    \bE \left[F^{uk} (W)^{\otimes 2}\right]
    .
\end{align}

First, consider the case of $\cS = \cT.$
It holds that $F^p = F^k$ because it holds $c_j(\beta_j^\ast, e_j^\ast, \gamma^\ast) = 0,$ since
\begin{align*}
    \bE \left[
        e_j^\ast (X; \beta_j^\ast)
        \sum_{i \in \cS}
        D_i
        S_i^\ast (X)^\prime
    \right]
    =
    \bE \left[
        e_j^\ast (X; \beta_j^\ast)
        \left(
            \frac{\partial}{\partial \gamma} 
            \at{\gamma = \gamma^\ast}
            \sum_{i \in \cT}
            p_i (X; \gamma)
        \right)^\prime
    \right]
    =
    0
    .
\end{align*}
We also have $F^{uk} = F^p,$ as $D_{\cS} - p_{\cS}^{\ast} (X) = 0.$
These observations lead to a well-known result: the relationship (\ref{eq:comparison-baseline}) holds all with equality.
In other words, knowing the propensity score does not improve the efficiency, as discovered by \cite{hahn1998role} and \cite{lee2018efficient}.

In what follows, we assume $\cS \subsetneq \cT$ unless stated otherwise.
The connection between $F^p$ and $F^{uk}$ can be understood geometrically.
\begin{proposition} \label{prop:projection-of-IF}
The function $F^p$ is the $L^2$-projection of $F^{uk},$ that is, for each $j \in \cT,$ the unique solution of
\begin{align*}
    \min_{c_j \in \bR^{d_m \times d_\gamma}}
    \bE \left[
        \norm{
            F_j^{uk} (W)
            -
            \left(
                D_j
                s_j(W; \beta_j^\ast, e_j^\ast, \gamma^\ast)
                +
                c_j
                \sum_{i \in \cT} 
                D_i 
                S_i^\ast(X)
                +
                t_j(W; \beta_j^\ast, e_j^\ast, \gamma^\ast)
            \right)
        }^2
    \right]
\end{align*}
is given by $c_j = c_j (\beta_j^\ast, e_j^\ast, \gamma^\ast).$
\end{proposition}

The orthogonality between $F^p$ and $F^{uk} - F^p$ admits the decomposition
\begin{align*}
    \bE \left[F^{uk} (W)^{\otimes 2}\right]
    =
    \bE \left[F^p (W)^{\otimes 2}\right]
    +
    \bE \left[(F^{uk} (W) - F^p (W))^{\otimes 2}\right]
    .
\end{align*}
Hence, if the minimization problem of Proposition \ref{prop:projection-of-IF} does not attain zero, which is the case in most cases, then the second inequality of (\ref{eq:comparison-baseline}) strictly holds.

\section{Efficiency Gains in Stratified Experiments} \label{sec:stratified-experiment}
In this section, we take a quantitative approach to investigate how much a parametric restriction on the propensity score improves the efficiency bound of parameters of interest in stratified experiments.
When the propensity score is known to be stratified, the efficient influence function and the semiparametric efficiency bound turn out to have closed forms.
After deriving analytic formulas, we discuss efficiency gains from knowing the restriction on the propensity score.

Let us first define the class of stratified propensity scores.
To do so, fix a measurable partition $\cX = \bigsqcup_{k = 1}^K \cX_k$ of the space of covariate vectors.
Consider the following class of propensity scores
\begin{align*}
    p_j(x)
    =
    \sum_{k = 1}^K p_{j, k} \bI\{x \in \cX_k\}
\end{align*}
for some constants $p_{j, k}.$
We call this form of propensity score a $K$-stratified propensity score with partition $\bigsqcup_{k = 1}^K \cX_k.$
This class is parameterized by $(JK)$-parameters that live in
\begin{align*}
    \left\{
        (p_{j, k})_{1 \leq j \leq J, 1 \leq k \leq K}
        \ \bigg | \
        p_{j, k}
        \geq 
        0
        ,
        \
        \sum_{j = 1}^J p_{j, k}
        \leq
        1
    \right\}
    .
\end{align*}
The probabilities for $j = 0$ are specified as $p_{0, k} = 1 - \sum_{j = 1}^J p_{j, k}$ for $k = 1, \cdots, K.$
Let $p_j^\ast (x) = \sum_{k = 1}^K p_{j, k}^\ast \bI\{X \in \cX_k\}$ be the true propensity score.
We denote $p_{\cS, k}^\ast \coloneqq \sum_{j \in \cS} p_{j, k}^\ast.$

The class of stratified propensity scores covers many important situations.
For example, it includes cases where the propensity score is known to depend only on categorical covariates such as race, gender, month of birth, city of residence, and so on.
Even when covariates are continuous, people often discretize them in practice, which also falls within the scope of the current setting.

As in Section \ref{subsec:comparison}, let $V^k, V^p$ and $V^{uk}$ be the efficiency bounds of $\beta$ when the propensity score is known, known to be $K$-stratified by partition $\bigsqcup_{k = 1}^K \cX_k,$ and unknown, respectively.
We are interested in the relationship among these bounds, but thanks to (\ref{eq:derivative-not-matter}), we focus on the comparison among $\bE \left[F^k(W)^{\otimes 2}\right],$ $\bE \left[F^p (W)^{\otimes 2}\right],$ and $\bE \left[F^{uk} (W)^{\otimes 2}\right].$
We regard $\Delta_K^{uk \to p} \coloneqq \bE \left[F^{uk} (W)^{\otimes 2}\right] - \bE \left[F^p (W)^{\otimes 2}\right]$ as the efficiency gain from knowing that the propensity score is $K$-stratified, and similarly, $\Delta_K^{p \to k} \coloneqq \bE \left[F^p (W)^{\otimes 2}\right] - \bE \left[F^k (W)^{\otimes 2}\right]$ is the value of knowing the propensity score on each subclass.
The stratification structure allows us to obtain analytical formulas of these quantities.
\begin{theorem} \label{thm:formula-efficiency-gain}
It holds that
\begin{align*}
    \Delta_K^{uk \to p}
    =
    \frac{1}{(p_{\cS}^{\ast})^2}
    \sum_{k = 1}^K
    p_{\cS, k}^\ast
    (1 - p_{\cS, k}^\ast)
    \bP (X \in \cX_k)
    \bV \left(
        \begin{pmatrix}
            e_0^\ast (X; \beta_0^\ast) \\
            \vdots \\ 
            e_J^\ast (X; \beta_J^\ast)
        \end{pmatrix} 
        \ \bigg | \
        X \in \cX_k
    \right)
    ,
\end{align*}
and
\begin{align*}
    \Delta_K^{p \to k}
    =
    \frac{1}{(p_{\cS}^{\ast})^2}
    \sum_{k = 1}^K
    p_{\cS, k}^\ast
    (1 - p_{\cS, k}^\ast)
    \bP (X \in \cX_k)
    \left\{
        \bE \left[
            \begin{pmatrix}
                e_0^\ast (X; \beta_0^\ast) \\
                \vdots \\ 
                e_J^\ast (X; \beta_J^\ast)
            \end{pmatrix} 
            \ \bigg | \
            X \in \cX_k
        \right]
    \right\}^{\otimes 2}
    .
\end{align*}
\end{theorem}

\begin{remark}
The formulas in Theorem \ref{thm:formula-efficiency-gain} are explicit in the sense that the matrix inversion that appears in $c_j (\beta_j^\ast, e_j^\ast, \gamma^\ast)$ is resolved.
This happens only for the stratified propensity score and cannot be extended to other parametric models.
The reason is as follows.
Recall that we are interested in the inverse of 
\begin{align*}
    \bE \left[
        \left(
            \sum_{i \in \cT}
            D_i
            S_i^\ast (X)
        \right)^{\otimes 2}
    \right]
    =
    \sum_{i \in \cT}
    \bE \left[
        D_i
        S_i^\ast (X)^{\otimes 2}
    \right]
    .
\end{align*}
In general, it is difficult to find an explicit formula for the inverse of the sum of matrices, like the one in the RHS, and this causes a problem in most parametric models.
For the stratified propensity score, however, the sum can be nicely decomposed into diagonal matrices, and it turns out to be explicitly invertible using the Woodbury formula.
This calculation exploits the property of stratified experiments that the parameter $p_{j, k}$ affects the propensity score only on the $k$th cell $\cX_k.$
See the Supplementary Material for more technical details.
\end{remark}

As we have discussed the case of $\cS = \cT$ in Section \ref{subsec:comparison}, we focus on $\cS \subsetneq \cT$ in what follows.
Notice that a $1$-stratified propensity score coincides with a constant propensity score, in which cases treatments are assigned completely at random.
As the two formulas in Theorem \ref{thm:formula-efficiency-gain} suggest, for $K = 1,$ it holds
\begin{align*}
    \Delta_1^{uk \to p}
    =
    \frac{1 - p_{\cS}^{\ast}}{p_{\cS}^{\ast}}
    \bE \left[
        \begin{pmatrix}
            e_0^\ast (X; \beta_0^\ast) \\
            \vdots \\ 
            e_J^\ast (X; \beta_J^\ast)
        \end{pmatrix}^{\otimes 2} 
    \right]
    >
    0
    ,
    \ \ \text{and} \ \
    \Delta_1^{p \to k}
    =
    0
    .
\end{align*}
An implication of this result is that the knowledge that the propensity score is constant over the whole covariate space strictly improves the efficiency, i.e., $\Delta_1^{uk \to p} > 0,$ but full information of the propensity score does not bring any additional value, i.e., $\Delta_1^{p \to k} = 0.$
This is not the case for $K > 1,$ because $\Delta_K^{p \to k}$ is strictly positive in general. In other words, knowing the specific values of the propensity score does improve the efficiency.
These observations generalize \cite{hahn1998role} who finds a similar fact in binary treatment cases.

In the stratified experiment setup, the value $\Delta_K^{uk \to p}$ of imposing a parametric restriction captures the \textit{in-class variance} while the additional value $\Delta_K^{p \to k}$ of pinning down the propensity score embodies the \textit{between-class bias}.
Too see this, suppose first that we know that the true propensity score is $K$-stratified by a given partition.
In this case, data in each subclass can be thought of as sample from a random assignment.
It holds that
\begin{align*}
    \bE [m (Y_j; \beta_j^\ast) \mid T \in \cS, X \in \cX_k]
    =
    \bE [m (Y_j; \beta_j^\ast) \mid X \in \cX_k]
    =
    \bE [e_j^\ast (X; \beta_j^\ast) \mid X \in \cX_k]
    ,
\end{align*}
that is, the moment for subpopulation in each subclass can be estimated based on all data points including observations to which the treatments in $\cS$ are not assigned.
Since all information in each subclass is exploited in this way, the knowledge of the partition removes the in-class efficiency loss.
Note that \cite{frolich2004note} provides a similar observation for binary average treatment effects on the treated.

Even if the functional form of the propensity score is known, there may be another source of efficiency loss that we call the between-class bias.
To see this, notice that
\begin{align*}
    \bE [m (Y_j; \beta_j^\ast) \mid T \in \cS]
    =
    \sum_{k = 1}^K
    \left(
        \frac{p_{\cS, k}^\ast}{p_{\cS}^{\ast}}
        \bP (X \in \cX_k)
    \right)
    \bE [e_j^\ast (X; \beta_j^\ast) \mid X \in \cX_k]
    .
\end{align*}
This decomposition implies that the moment function of interest is an aggregation of local conditional expectations in each subclass with a weight depending on the propensity score.
Therefore, one cannot reproduce the moment function efficiently without knowing the exact values of the propensity score, which leads to efficiency loss.
There is no loss in the absence of knowledge of the propensity score only when $\bE [e_j^\ast (X; \beta_j) \mid X \in \cX_k]$ is constant or $K = 1.$
The former case does not hold generically. 
The latter is consistent with the result from \cite{hahn1998role}, who shows that knowing that the propensity score is constant is the same as knowing its value in terms of efficiency, but our result also indicates that $K = 1$ is just a special case in that the between-bias is zero despite the fact that the exact value of the propensity score is unavailable.

At the end of this section, we consider the behavior of the efficiency gain $\Delta_K^{uk \to p}$ from knowing the partition when it is infinitely fine.
Intuitively, if cells of the partition are tiny, just knowing the partition is almost useless because the in-class variance of each cell is small.
This observation is justified in the following proposition by considering a nested sequence of partitions.
Assume $\cX$ is a subset of a Euclidean space, and let $\mathrm{vol} (\cdot)$ be the Lebesgue measure on it.

\begin{proposition} \label{prop:limit-stratified-propsnsoty-score-fine}
Let $K (n)$ be an increasing sequence of natural numbers.
For each $n \geq 1,$ let $\{\cX_k^{(n)}\}_{k = 1}^{K (n)}$ be a partition of $\cX$ satisfying the following condition:
for any $n \geq 2$ and $k \in \{1, \dots, K (n)\},$ there uniquely exists $k^\prime \in \{1, \cdots, K (n - 1)\}$ such that $\cX_{k^\prime}^{(n - 1)} \supset \cX_k^{(n)}$ and $\lim_{n \to \infty} \max_{1 \leq k \leq K (n)} \mathrm{vol} (\cX_k^{(n)}) = 0.$
Then, it holds that $\lim_{n \to \infty} \Delta_{K (n)}^{uk \to p} = 0.$
\end{proposition}

On the other hand, if the partition is known to be coarse, the efficiency is improved, as the following proposition claims.

\begin{proposition} \label{prop:limit-stratified-propsnsoty-score-coarse}
In the same setup as Proposition \ref{prop:limit-stratified-propsnsoty-score-fine}, assume $\lim_{n \to \infty} \max_{1 \leq k \leq K (n)} \mathrm{vol} (\cX_k^{(n)}) > 0.$
Then, it holds that $\lim_{n \to \infty} \Delta_{K (n)}^{uk \to p} > 0$ if $e_j^\ast (\cdot; \beta_j^\ast)$ is not constant on any open ball in $\cX.$
\end{proposition}

The assumption of Proposition \ref{prop:limit-stratified-propsnsoty-score-coarse} describes the situation where it is known that there is a unignorable region where assignments are completely at random.
In this case, the propensity score is restricted to the space of functions that are constant on the region, which leads to an efficiency gain.

Propositions \ref{prop:limit-stratified-propsnsoty-score-fine} and \ref{prop:limit-stratified-propsnsoty-score-coarse} are shown from the formulas that are derived in Theorem \ref{thm:formula-efficiency-gain}, but they can also be proven using the general theory in the next section.
See the Supplementary Material for details in this direction.

\section{Efficiency Gains from Large Parametric Models} \label{sec:param-vs-nonparam}
\subsection{General Theory}
We have shown in Section \ref{sec:speb} that imposing a parametric restriction on the propensity score weakly improves estimation efficiency, but as shown in Propositions \ref{prop:limit-stratified-propsnsoty-score-fine} and \ref{prop:limit-stratified-propsnsoty-score-coarse}, the amount of the efficiency gain varies depending on the parametric model assumed.
In order to see what kinds of parametric restrictions are valuable in terms of efficiency, it is helpful to compute the reduction $\bE \left[F^{uk} (W)^{\otimes 2}\right] - \bE \left[F^{p} (W)^{\otimes 2}\right]$ in the efficiency bound as in Section \ref{sec:stratified-experiment}.
Unfortunately, however, it is hard to find an explicit formula of this quantity in general because it involves a convoluted matrix inversion in $\bE \left[F^{p} (W)^{\otimes 2}\right].$
In this section, we take a more conservative approach, instead of calculating it explicitly.
Specifically, we investigate how the amount of the efficiency improvement changes as the parametric model becomes infinitely ``large.''
To formally study asymptotics for parametric models, we first define a countably infinite-dimensional model, which is thought of as the ``structure'' of the parametric model, and then restrict it to obtain a sequence of finite-dimensional parametric submodels in a way that they are nested.
We consider efficiency bounds along this sequence and see whether their limit attains the bound for the unknown propensity score.
If it does, then it implies that imposing the parametric model is almost the same as estimating the propensity score nonparametrically, especially when $n$ is large, so that the parametric restriction does not help significantly in terms of efficiency.
If the nonparametric bound is not attained even asymptotically, on the other hand, it means that the parametric structure intrinsically restricts the space of possible propensity scores.

As we have seen in the previous section, when $\cS = \cT,$ there is no efficiency gain from knowing the propensity score. Since the object of interest here is the difference $\bE \left[F^{uk} (W)^{\otimes 2}\right] - \bE \left[F^{p} (W)^{\otimes 2}\right],$ we assume $\cS \subsetneq \cT$ throughout this section. We also assume $((e_0^\ast)^\prime, \dots, (e_J^\ast)^\prime)^\prime \not \equiv 0;$ otherwise, the knowledge about the propensity score does not matter either.

To describe situations where the number of parameters grows with the ``structure'' of a parametric model fixed, we consider a nested family of parametric models that is induced by a large base model.
Let $\Theta^\bN$ be a countably infinite-dimensional parameter space where $\Theta \subset \bR,$ which defines a base model of the propensity score, $\cM_\infty^p \coloneqq \left\{(p_j (\cdot; \gamma))_{j \in \cT} \mid \gamma \in \Theta^\bN\right\}.$
Let $d (n) \nearrow \infty$ be an increasing sequence.
Define a $d (n)$-dimensional parametric submodel of $\cM_\infty^p$ as
\begin{align*}
    \cM_n^p
    \coloneqq
    \left\{
        (p_j (\cdot; \gamma))_{j \in \cT}
        \mid
        \gamma \in \Theta^\bN
        , \ 
        k \geq d (n) + 1
        \Rightarrow
        \gamma_k = 0
    \right\}
    .
\end{align*}
That is, the submodel $\cM_n^p$ is parameterized by the first $d (n)$ elements of $\gamma,$ which is denoted by $\gamma^{(n)} = (\gamma_1, \dots, \gamma_{d (n)}).$
It is obvious that these models are nested in the sense that
\begin{align*}
    \cM_{1}^p
    \subset
    \cM_{2}^p
    \subset
    \dots
    \subset
    \cM_\infty^p
    \subset
    \cM^{uk}
    ,
\end{align*}
where $\cM^{uk}$ is the set of all possible propensity scores.
Moreover, assume that the true propensity score $(p_j^\ast)_{j \in \cT} = (p_j(\cdot; \gamma^\ast))_{j \in \cT}$ is of finite dimension, i.e., $\gamma^\ast = (\gamma_1^\ast, \dots, \gamma_{d (N^\ast)}^\ast, 0, 0, \dots)$ for some $N^\ast.$

Let $V^{p, n}$ be the semiparametric efficiency bound in Theorem \ref{thm:eff-bound} when the propensity score is assumed to lie in $\cM_n^p.$
Then, $V^{p, n}$ is nondecreasing in $n$ in the sense of positive semidefinite matrix since $\cM_n^p$ is nondecreasing.
It is also bounded above by $V^{uk},$ which is the semiparametric efficiency bound for $\cM^{uk},$ so that the sequence $V^{p, n}$ converges to some positive definite matrix denoted by $V^{p, \infty}.$ 
This is shown in the Supplementary Material.
It holds in general that $V^{p, n} \nearrow V^{p, \infty} \leq V^{uk}.$

The score of the parametric model $\cM_n^p$ is
\begin{align*}
    S_j^{\ast n} (x)
    \coloneqq
    \frac{\partial}{\partial \gamma^{(n)}}
    \at{\gamma = \gamma^\ast}
    \log p_j (x; \gamma)
    .
\end{align*}
Also, define a $d (n) \times J$ matrix valued function
\begin{align*}
    \bS_n
    \coloneqq
    \left(
        \frac{\partial}{\partial \gamma_i} 
        \at{\gamma = \gamma^\ast}
        \log p_j (\cdot; \gamma)
    \right)_{1 \leq i \leq d (n), 1 \leq j \leq J}
    =
    \begin{pmatrix}
        S_1^{\ast n} & \dots & S_J^{\ast n}
    \end{pmatrix}
    .
\end{align*}

For $j \in \cT$ and $\ell \in \{1, \dots, d_m\},$ let
\begin{align*}
    \bar \ve_{j, \ell} (X)
    \coloneqq
    \begin{pmatrix}
        (\bI\{1 \in \cS\} - p_{\cS}^{\ast} (X)) e_{j, \ell}^\ast (X; \beta_j^\ast)
        &
        \dots
        &
        (\bI\{J \in \cS\} - p_{\cS}^{\ast} (X)) e_{j, \ell}^\ast (X; \beta_j^\ast)
    \end{pmatrix}
\end{align*}
where $e_j^\ast = (e_{j, 1}^\ast, \dots, e_{j, d_m}^\ast)^\prime.$
Note that $\bar \ve_{j, \ell}$ is a vector consisting of at most two functions.
In particular, its entries indexed by $i \in \cS$ have $(1 - p_{\cS}^{\ast} (\cdot)) e_{j, \ell}^\ast (\cdot; \beta_j^\ast),$ and the others have $- p_{\cS}^{\ast} (\cdot) e_{j, \ell}^\ast (\cdot; \beta_j^\ast).$

Now, we introduce the following condition.
\begin{conditionp}{F} \label{cond:score-is-flexible}
For any $j \in \cT$ and $\ell \in \{1, \dots, d_m\},$ it holds that
\begin{align*}
    \lim_{n \to \infty}
    \inf_{c_n \in \bR^{d (n)}}
    \norm{c_n^\prime \bS_n - \bar \ve_{j, \ell}}_{(L^2(F_X))^J}
    =
    0
    .
\end{align*}
\end{conditionp}

As we see in the Supplementary Material, the two functions $(1 - p_{\cS}^{\ast} (\cdot)) e_{j, \ell}^\ast (\cdot; \beta_j^\ast)$ and $- p_{\cS}^{\ast} (\cdot) e_{j, \ell}^\ast (\cdot; \beta_j^\ast),$ which form $\bar \ve_{j, \ell},$ appear in the efficient influence function of $\beta$ under the nonparametric model $\cM^{uk}.$
It is also shown that the function $c_n^\prime \bS_n$ is the counterpart of $\bar \ve_{j, \ell}$ under the parametric model $\cM_n^p.$
Therefore, Condition \ref{cond:score-is-flexible} implies that the sequence of efficient influence functions of the increasing parametric models converges to that of the nonparametric model.

It is intuitive that the efficiency bound of the limit parametric model $\cM_\infty^p$ is equal to that of the nonparametric model if the parametric model is sufficiently ``flexible.''
Condition \ref{cond:score-is-flexible} can be thought of as a criterion of such flexibility and is shown to characterize when a parametric restriction on the propensity score does \textit{not} improve the efficiency bound asymptotically.
\begin{theorem} \label{thm:many-param-limit}
Condition \ref{cond:score-is-flexible} holds if and only if $V^{p, \infty} = V^{uk}.$
\end{theorem}

Theorem \ref{thm:many-param-limit} tells us that knowledge of a parametric restriction on the propensity score that violates Condition \ref{cond:score-is-flexible} is beneficial in terms of efficiency because it improves the efficiency bound even asymptotically, i.e., $V^{p, \infty} < V^{uk}.$
It also implies that as long as Condition \ref{cond:score-is-flexible} is satisfied, it may hold that $V^{p, \infty} = V^{uk},$ even though $\cM_\infty^p \subsetneq \cM^{uk}.$
In other words, restricting the space of propensity scores does not necessarily improve estimation efficiency.

\subsection{Application to Logistic Propensity Scores}
In this subsections, we give two examples that satisfy or violate Condition \ref{cond:score-is-flexible} for logistic propensity scores.

\medskip

\noindent
\textit{Full-rank logistic propensity scores.}
Consider a parameter
\begin{align*}
    \Gamma 
    = 
    \begin{pmatrix}
        \Gamma_1 & \Gamma_2 & \dots
    \end{pmatrix}
    \in 
    (\bR^J)^\bN
     \ \text{where} \ 
    \Gamma_n 
    =
    (\Gamma_{n, 1}, \dots, \Gamma_{n, J})^\prime
    \in \bR^J
    .
\end{align*}
Let $(b_n)_{n \in \bN} \in (L^2 (F_X))^\bN$ be a dictionary of functions of covariate $x.$
Let 
\begin{align*}
    \gamma 
    = 
    \mathrm{vec}(\Gamma)
    =
    \begin{pmatrix}
        \Gamma_1 \\
        \Gamma_2 \\
        \vdots
    \end{pmatrix}
    .
\end{align*}
The object of interest here is a full-rank logistic propensity score that is defined as
\begin{align*}
    p_j (x; \gamma)
    =
    \frac{
        \exp \left(
            \sum_{k = 1}^\infty \Gamma_{k, j} b_k (x)
        \right)
    }{
        1 
        + 
        \sum_{i = 1}^J 
        \exp \left(
            \sum_{k = 1}^\infty \Gamma_{k, i} b_k (x)
        \right)
    }
\end{align*}
for $1 \leq j \leq J,$ and $p_0 (x; \gamma) = 1 - \sum_{j = 1}^J p_j (x; \gamma).$
In this example, the base model is $\cM_\infty^p = \{(p_j (\cdot; \gamma))_{j \in \cT} \mid \gamma \in \bR^\bN\}.$
By restricting the space of $\gamma,$ we can obtain a nested sequence of finite-dimensional parametric models as follows:
\begin{align*}
    \cM_n^p
    &=
    \left\{
        (p_j (\cdot; \gamma))_{j \in \cT}
        \in
        \cM_\infty^p
        \mid
        k \geq n + 1
        \Rightarrow
        \Gamma_k = 0
    \right\}
    \\
    &=
    \left\{
        (p_j (\cdot; \gamma))_{j \in \cT}
        \in
        \cM_\infty^p
        \ \bigg | \ 
        \gamma 
        =
        \mathrm{vec}(\Gamma)
        \ \text{where} \ 
        \Gamma
        =
        \begin{pmatrix}
            \Gamma_1 & \dots & \Gamma_n & 0 & \dots
        \end{pmatrix}
    \right\}
    .
\end{align*}
Note that $\cM_n^p$ is $d (n)$-dimensional where $d (n) = nJ.$
Suppose that the true propensity score is of order $N^\ast,$ i.e., $(p_j^\ast)_{j \in \cT} \in M_{N^\ast}^p$ with the true parameter $\Gamma^\ast = \begin{pmatrix} \Gamma_1^\ast & \dots & \Gamma_{N^\ast}^\ast & 0 & \dots \end{pmatrix}$ and the corresponding $\gamma^\ast.$
The following result provides a necessary and sufficient condition for the bound for this parametric model to approach the nonparametric bound.
\begin{proposition} \label{prop:full-logit}
For a full-rank logistic propensity score, $V^{p, \infty} = V^{uk}$ if and only if the linear span of $(b_n)_{n \in \bN}$ contains 
\begin{align*}
    \frac{\bI\{i \in \cS\} - p_{\cS}^{\ast} (X)}{1 - p_i^\ast (X)}
    e_{j, \ell}^\ast (X; \beta_j^\ast)    
\end{align*}
for all $i, j \in \cT$ and $\ell \in \{1, \dots, d_m\}.$
\end{proposition}
An implication of this proposition is that if the true propensity score is indeed full-rank logistic with $(b_n)_{n \in \bN}$ that spans $L^2 (F_X),$ knowing this does not improve the efficiency bound when the parametric model has many parameters.
In other words, such a parametric restriction provides no efficiency gain asymptotically.
However, the condition of the statement does not require $b_n$ to span $L^2 (F_X),$ in which case, the parametric model of the propensity score is restrictive compared to the nonparametric one.
Proposition \ref{prop:full-logit} claims that knowing unignorable aspects of the propensity score does not necessarily improve the efficiency.

Note that Proposition \ref{prop:full-logit} is a similar result to Proposition \ref{prop:limit-stratified-propsnsoty-score-fine}.
Both propositions tell us that if the propensity score is known to lie in a large parametric model, estimating parameters of interest in the parametric model is as hard as doing so nonparametrically.

\medskip

\noindent
\textit{Degenerate logistic propensity scores.}
As an example that violates Condition \ref{cond:score-is-flexible}, we define a degenerate logistic propensity score as
\begin{align*}
    p_j (x; \tilde \gamma)
    =
    \frac{
        \exp(\sum_{k = 1}^\infty \tilde \gamma_k b_k (x))
    }{1 + \sum_{i = 1}^J \exp(\sum_{k = 1}^\infty \tilde \gamma_k b_k (x))}
    =
    \frac{
        \exp(\sum_{k = 1}^\infty \tilde \gamma_k b_k (x))
    }{1 + J \exp(\sum_{k = 1}^\infty \tilde \gamma_k b_k (x))}
\end{align*}
for $j = 1, \dots, J$ where $\tilde \gamma \in \bR^\bN.$
This model is obtained by restricting the parameter space of the full-rank logistic propensity score by setting $\Gamma_{n, j} = \tilde \gamma_n.$
Under this model, the propensity scores for $j = 1, \dots, J$ are all identical.
\begin{proposition} \label{prop:deg-logit}
For a degenerated logistic propensity score, $V^{p, \infty} < V^{uk}$ if (i) $\cS \cap \{1, \dots, J\} \neq \emptyset$ and (ii) $\{1, \dots, J\} \setminus \cS \neq \emptyset.$
\end{proposition}
To understand the assumption of the proposition, recall that $\emptyset \subsetneq \cS \subsetneq \cT,$ and assume $0 \in \cS$ for simplicity.
Then, condition (ii) is automatically satisfied; otherwise $\cS = \cT.$
Condition (i) requires $\cS$ to have at least one element other than $0.$
In the case of $\cT = \{0, 1, 2\},$ this is satisfied, for example, when we are interested in treatment effects on those who are assigned to $\cS = \{0, 1\}.$
Notice that condition (i) is never satisfied for binary treatments, $\cT = \{0, 1\}$ unless $\cS = \cT.$
Indeed, $V^{p, \infty} = V^{uk}$ holds for such cases, if the dictionary $(b_n)_{n \in \bN}$ spans $L^2 (F_X).$
See the Supplementary Material for further details.

Proposition \ref{prop:deg-logit} implies that under the obviously verifiable conditions, the space of degenerated logistic propensity scores is restrictive, and consequently, assuming this structure improves the estimation efficiency even asymptotically.

\appendix

\section{Proofs} \label{app:proof}

\subsection{Proof of Theorem \ref{thm:eff-bound}} \label{sec:proof-speb}
\begin{proof}
We trace the approach given by \cite{bickel1993efficient} and \cite{hahn1998role}.
We consider a semiparametric model $\cP = \{q_{\{f_j\}_{j \in \cT}, \gamma, f}\}$ such that
\begin{align*}
    q_{\{f_j\}_{j \in \cT}, \gamma, f}(y, t, x)
    =
    \left\{
        \prod_{j \in \cT}
        \left(
            f_j (y \mid x) p_j (x; \gamma)
        \right)^{\bI\{t = j\}}
    \right\}
    f(x)
\end{align*}
where $f_j(\cdot \mid x)$ is a density on $\bR$ for each $x \in \mathcal{X},$ and $f$ is a density on $\mathcal{X}.$
Consider a regular parametric submodel
\begin{align*}
    q_\theta(y, t, x)
    =
    \left\{
        \prod_{j \in \cT}
        \left(
            f_{j, \theta} (y \mid x) p_j (x; \gamma_\theta)
        \right)^{\bI\{t = j\}}
    \right\}
    f_\theta(x)
\end{align*}
with the true value $\theta = \theta^\ast$ and the corresponding score is
\begin{align*}
    s_\theta(y, t, x)
    =
    \sum_{j \in \cT}
    d_j s_{j, \theta}(y \mid x)
    +
    \sum_{j \in \cT}
    d_j \frac{\dot p_{j, \theta} (x)}{p_j^\ast (x)}
    +
    t_\theta(x)
\end{align*}
where
\begin{align*}
    s_{j, \theta}(y \mid x)
    &\coloneqq
    \frac{\partial}{\partial \theta} \at{\theta = \theta^\ast} \log f_{j, \theta}(y \mid x)
    ,
    \\
    \dot p_{j, \theta} (x)
    &\coloneqq
    \frac{\partial}{\partial \theta} \at{\theta = \theta^\ast} p_j (x; \gamma_\theta)
    =
    \left(\frac{\partial \gamma_\theta}{\partial \theta} \at{\theta = \theta^\ast} \right)^\prime
    S_j^\ast (x)
    p_j^\ast (x)
    ,
    \\
    t_\theta(x)
    &\coloneqq
    \frac{\partial}{\partial \theta} \at{\theta = \theta^\ast} \log f_{\theta}(x)
    .
\end{align*}
Thus, the tangent space of $\cP$ at the true parameter is $\dot \cP = \cH_1 \oplus \cH_2 \oplus \cH_3$ where
\begin{align*}
    \cH_1
    &\coloneqq
    \left\{
        w
        \mapsto
        \sum_{j \in \cT}
        d_j s_j (y \mid x)
        \ \Big | \ 
        s_j \in L^2_0 (F_{Y_j \mid X})
    \right\}
    ,
    \\
    \cH_2
    &\coloneqq
    \left\{
        w
        \mapsto
        c^\prime
        \sum_{j \in \cT}
        d_j
        S_j^\ast (x)
        \ \Big | \ 
        c \in \bR^{d_\gamma}
    \right\}
    ,
    \\
    \cH_3
    &\coloneqq
    L^2_0 (F_X)
    ,
\end{align*}
where $w = (y, (d_j)_{j \in \cT}, x),$ $F_{Y_j \mid X}$ is the distribution of $Y_j$ conditional on $X,$ and $L_0^2$ denotes the space of square integrable functions with mean zero.

The moment condition implies that for any $d_\beta \times d_m$ matrix $A,$ it holds
\begin{align*}
    A
    \bE \left[
        \begin{pmatrix}
            m (Y_0; \beta_0) \\
            \vdots \\
            m (Y_J; \beta_J)
        \end{pmatrix}
        \ \bigg | \
        T \in \cS
    \right]
    =
    0
    .
\end{align*}
On the parametric submodel, it holds
\begin{align*}
    A
    \bE_\theta \left[
        \begin{pmatrix}
            m (Y_0; \beta_0 (\theta)) \\
            \vdots \\
            m (Y_J; \beta_J (\theta))
        \end{pmatrix}
        \ \bigg | \
        T \in \cS
    \right]
    =
    0
    .
\end{align*}
Differentiating both sides in $\theta,$ we obtain
\begin{align*}
    \frac{\partial}{\partial \theta} \at{\theta = \theta^\ast} \beta (\theta)
    =
    -
    (A \cJ)^{-1}
    A
    \left(
        \frac{\partial}{\partial \theta}
        \bE_\theta \left[
            \begin{pmatrix}
                m (Y_0; \beta_0^\ast) \\
                \vdots \\
                m (Y_J; \beta_J^\ast)
            \end{pmatrix}
            \ \bigg | \
            T \in \cS
        \right]
    \right)
    .
\end{align*}
For any regular parametric submodel with score $s_\theta \in \dot \cP$, it holds that
\begin{align*}
    \frac{\partial}{\partial \theta} \at{\theta = \theta^\ast} \bE_\theta [m (Y_j; \beta_j^\ast) \mid T \in \cS]
    =
    \bE [F_j^p (W) s_\theta (W)]
\end{align*}
for each $j.$
Thus, we have
\begin{align*}
    \frac{\partial}{\partial \theta} \at{\theta = \theta^\ast} \beta (\theta)
    =
    -
    (A \cJ)^{-1}
    A
    \bE \left[
        F^p (W)
        s_\theta (W)
    \right]
    .
\end{align*}
Substituting $A = \cJ^\prime \bE \left[F^p (W)^{\otimes 2}\right]^{-1}$ gives the efficient influence function 
\begin{align*}
    \psi^p (W) 
    = 
    - 
    \left(
    \cJ^\prime 
        \bE \left[
            F^p (W)^{\otimes 2}
        \right]^{-1} 
        \cJ
    \right)^{-1} 
    \cJ^\prime 
    \bE \left[
        F^p (W)^{\otimes 2}
    \right]^{-1} 
    F^p (W)
    .
\end{align*}
The efficiency bound is, therefore, $V^p = \left(\cJ^\prime \bE \left[F^p (W)^{\otimes 2}\right]^{-1} \cJ\right)^{-1}.$
\end{proof}

\subsection{Proof of Proposition \ref{prop:projection-of-IF}}
\begin{proof}
The solution coincides with that of
\begin{align*}
    \min_{c_j \in \bR^{d_m \times d_\gamma}}
    \bE \left[
        \norm{
            \frac{1}{p_{\cS}^{\ast}}
            (D_{\cS} - p_{\cS}^{\ast} (X))
            e_j^\ast (X; \beta_j^\ast)
            -
            c_j
            \sum_{i \in \cT} D_i S_i^\ast (X)
        }^2
    \right]
    .
\end{align*}
By the standard formula of least squares, the objective is minimized at
\begin{align*}
    c_j
    &=
    \bE \left[
        \frac{1}{p_{\cS}^{\ast}}
        (D_{\cS} - p_{\cS}^{\ast} (X))
        e_j^\ast (X; \beta_j^\ast)
        \sum_{i \in \cT} D_i S_i^\ast (X)^\prime
    \right]
    \bE \left[
        \left(
            \sum_{i \in \cT} D_i S_i^\ast (X)
        \right)^{\otimes}
    \right]^{-1}
    \\
    &=
    \frac{1}{p_{\cS}^{\ast}}
    \bE \left[
        e_j^\ast (X; \beta_j^\ast)
        \sum_{i \in \cS} D_i S_i^\ast (X)^\prime
    \right]
    \bE \left[
        \left(
            \sum_{i \in \cT} D_i S_i^\ast (X)
        \right)^{\otimes}
    \right]^{-1}
    \\
    &=
    c_j (\beta_j^\ast, e_j^\ast, \gamma^\ast)
    ,
\end{align*}
where the second equality follows since
\begin{align*}
    \bE \left[
        p_{\cS}^{\ast} (X) 
        e_j^\ast (X; \beta_j^\ast) 
        \sum_{i \in \cT} 
        D_i 
        S_i^\ast (X)
    \right] 
    = 
    \bE \left[
        p_{\cS}^{\ast} (X) 
        e_j^\ast (X; \beta_j^\ast) 
        \frac{\partial}{\partial \gamma}
        \at{\gamma = \gamma^\ast}
        \left(
            \sum_{i \in \cT} p_i (X; \gamma)
        \right)
    \right] 
    = 
    0
\end{align*}
holds.
\end{proof}

\subsection{Proof of Theorem \ref{thm:formula-efficiency-gain}} \label{app:proof-formula-efficiency-gain}

We first derive the efficient influence function by applying Theorem \ref{thm:eff-bound}.
\begin{theorem} \label{thm:speb-stratified}
For the $K$-stratified propensity score with partition $\bigsqcup_{k = 1}^K \cX_k,$ the efficient influence function is given by that in Theorem \ref{thm:eff-bound} with
\begin{align*}
    F_j^p (W)
    =
    &
    \frac{1}{p_{\cS}^{\ast}}
    \sum_{k = 1}^K
    \bI\{X \in \cX_k\}
    \Bigg\{
        \frac{p_{\cS, k}^{\ast}}{p_{j, k}^{\ast}}
        D_j
        (m (Y; \beta_j^\ast) - e_j^\ast (X; \beta_j^\ast))
        \\
        +
        &
        (D_{\cS} - p_{\cS, k}^{\ast})
        \bE \left[
            e_j^\ast (X; \beta_j^\ast) \mid X \in \cX_k
        \right]
        +
        p_{\cS, k}^{\ast}
        e_j^\ast (X; \beta_j^\ast)
    \Bigg\}
    ,
\end{align*}
where $p_{\cS, k}^{\ast} = \sum_{i \in \cS} p_{i, k}^{\ast}$ and $p_{\cS}^{\ast} = \sum_{k = 1}^K p_{\cS, k}^{\ast} \bP(X \in \cX_k).$
\end{theorem}

\begin{proof}[of Theorem \ref{thm:speb-stratified}]
Notice that
\begin{align*}
    p_i^\ast(x) S_i^\ast(x)
    =
    e_i^J 
    \otimes
    \begin{pmatrix}
        \bI\{x \in \cX_1\} \\
        \vdots \\
        \bI\{x \in \cX_K\}
    \end{pmatrix}
\end{align*}
for $1 \leq i \leq J$ and 
\begin{align*}
    p_0^\ast(x) S_0^\ast(x)
    =
    - 
    \iota^J 
    \otimes
    \begin{pmatrix}
        \bI\{x \in \cX_1\} \\
        \vdots \\
        \bI\{x \in \cX_K\}
    \end{pmatrix}
\end{align*}
where $e_i^J$ is the $J$-dimensional $i$th unit vector and $\iota^J$ is the $J$-dimensional vector of ones.
Then, we have
\begin{align*}
    \sum_{i \in \cT}
    D_i
    S_i^\ast (X)
    =
    \sum_{i = 1}^J 
    \sum_{k = 1}^K
    \left(\frac{D_i}{p_{i, k}^\ast} - \frac{D_0}{p_{0, k}^\ast}\right)
    \bI\{X \in \cX_k\}
    (e_i^J \otimes e_k^K)
    ,
\end{align*}
and therefore, 
\begin{align*}
    \bE \left[
        e_j^\ast (X; \beta_j^\ast)
        \sum_{i \in \cS}
        D_i
        S_i^\ast (X)^\prime
    \right]
    =
    -
    \sum_{i \in \cS^\complement}
    \sum_{k = 1}^K
    \bE [e_j^\ast (X; \beta_j^\ast) \mid X \in \mathcal{X}_k]
    \bP(X \in \mathcal{X}_k)
    (e_i^J \otimes e_k^K)
    .
\end{align*}
One can also verify 
\begin{align*}
    \bE\left[
        \left(\sum_{i \in \cT} D_i S_i^\ast (X)\right)^{\otimes 2}
    \right]
    =
    \sum_{i = 1}^J 
    \sum_{k = 1}^K
    \frac{\bP(X \in \cX_k)}{p_{i, k}^\ast}
    (E_{ii}^J \otimes E_{kk}^K)
    +
    \left(
        1^J
        \otimes 
        \sum_{k = 1}^K
        \frac{\bP(X \in \cX_k)}{p_{0, k}^\ast} 
        E_{kk}^K
    \right)
\end{align*}
where $E_{ij}^J \coloneqq e_i^J (e_j^J)^\prime,$ and $1^J$ is the $J \times J$ matrix consisting of ones.
To find its inverse, we use the following lemma.
\begin{lemma} \label{lem:inverse}
Let $C = \sum_{i = 1}^J \sum_{k = 1}^K c_{ik} (E_{ii}^J \otimes E_{kk}^K)$ and $D = \sum_{k = 1}^K d_k E_{kk}^K$ for $c_{ik}, d_k \neq 0.$
Then,
\begin{align*}
    (C + 1^J \otimes D)^{-1}
    =
    \sum_{i = 1}^J 
    \sum_{k = 1}^K 
    c_{ik}^{-1} (E_{ii}^J \otimes E_{kk}^K)
    -
    \sum_{i = 1}^J 
    \sum_{j = 1}^J 
    \sum_{k = 1}^K 
    d_k c_{ik}^{-1} c_{jk}^{-1} p_k^{-1}
    (E_{ij}^J \otimes E_{kk}^K)
    ,
\end{align*}
where
\begin{align*}
    p_k
    =
    1
    +
    \sum_{i = 1}^J
    \frac{d_k}{c_{ik}}
    .
\end{align*}
\end{lemma}
Using this lemma, we have
\begin{align*}
    \bE\left[
        \left(\sum_{i \in \cT} D_i S_i^\ast (X)\right)^{\otimes 2}
    \right]^{-1}
    =
    \sum_{i = 1}^J 
    \sum_{k = 1}^K 
    \frac{p_{i, k}^\ast}{\bP (X \in \cX_k)}
    (E_{ii}^J \otimes E_{kk}^K)
    -
    \sum_{i = 1}^J 
    \sum_{j = 1}^J 
    \sum_{k = 1}^K 
    \frac{p_{i, k}^\ast p_{j, k}^\ast}{\bP (X \in \cX_k)}
    (E_{ij}^J \otimes E_{kk}^K)
    .
\end{align*}
Thus, we obtain
\begin{align*}
    c_j(\beta_j^\ast, e_j^\ast, \gamma^\ast) 
    \sum_{i \in \cT} D_i S_i^\ast (X)
    &=
    \bE \left[
        e_j^\ast (X; \beta_j^\ast)
        \sum_{i \in \cS}
        D_i
        S_i^\ast (X)^\prime
    \right]
    \bE\left[
        \left(\sum_{i \in \cT} D_i S_i^\ast (X)\right)^{\otimes 2}
    \right]^{-1}
    \sum_{i \in \cT} D_i S_i^\ast (X)
    \\
    &=
    \frac{1}{p_{\cS}^{\ast}}
    \sum_{k = 1}^K
    (D_{\cS} - p_{\cS, k}^{\ast})
    \bI\{X \in \cX_k\}
    \bE [e_j^\ast (X; \beta_j^\ast) \mid X \in \cX_k]
    ,
\end{align*}
which completes the proof.
\end{proof}

\begin{proof} [of Lemma \ref{lem:inverse}]
Since $1^J \otimes D = (\iota^J \otimes D^{1/2}) (\iota^J \otimes D^{1/2})^\prime,$ we have
\begin{align*}
    (C + 1^J \otimes D)^{-1}
    &=
    C^{-1}
    -
    C^{-1}
    (\iota^J \otimes D^{1/2})
    P^{-1}
    (\iota^J \otimes D^{1/2})^\prime
    C^{-1}
\end{align*}
by the Woodbury formula, where
\begin{align*}
    P
    =
    I
    +
    (\iota^J \otimes D^{1/2})^\prime
    C^{-1}
    (\iota^J \otimes D^{1/2})
    =
    \sum_{k = 1}^K
    p_k
    E_{kk}^K
    .
\end{align*}
Then, one can also verify
\begin{align*}
    C^{-1}
    (\iota^J \otimes D^{1/2})
    P^{-1}
    (\iota^J \otimes D^{1/2})^\prime
    C^{-1}
    =
    \sum_{i = 1}^J 
    \sum_{j = 1}^J 
    \sum_{k = 1}^K 
    d_k c_{ik}^{-1} c_{jk}^{-1} p_k^{-1}
    (E_{ij}^J \otimes E_{kk}^K)
    ,
\end{align*}
which completes the proof.
\end{proof}

\begin{proof}[of Theorem \ref{thm:formula-efficiency-gain}]
By using the formulas (\ref{eq:EIF-known}) and (\ref{eq:EIF-nonparametric}), we obtain the efficient influence functions for the cases of known and unknown propensity scores as follows:
\begin{align*}
    F_j^k (W)
    &=
    \frac{1}{p_{\cS}^{\ast}}
    \sum_{k = 1}^K
    \bI\{X \in \cX_k\}
    \left\{
        \frac{p_{\cS, k}^{\ast}}{p_{j, k}^{\ast}}
        D_j
        (m (Y; \beta_j^\ast) - e_j^\ast (X; \beta_j^\ast))
        +
        p_{\cS, k}^{\ast}
        e_j^\ast (X; \beta_j^\ast)
    \right\}
    ,
    \\
    F_j^{uk} (W)
    &=
    \frac{1}{p_{\cS}^{\ast}}
    \sum_{k = 1}^K
    \bI\{X \in \cX_k\}
    \Bigg\{
        \frac{p_{\cS, k}^{\ast}}{p_{j, k}^{\ast}}
        D_j
        (m (Y; \beta_j^\ast) - e_j^\ast (X; \beta_j^\ast))
        \\
        &
        \hspace{15em}
        +
        (D_{\cS} - p_{\cS, k}^{\ast})
        e_j^\ast (X; \beta_j^\ast)
        +
        p_{\cS, k}^{\ast}
        e_j^\ast (X; \beta_j^\ast)
    \Bigg\}
    .
\end{align*}
Thus, we have
\begin{align*}
    \Delta_K^{uk \to p}
    &=
    \bE \left[F^{uk} (W)^{\otimes}\right]
    -
    \bE \left[F^p (W)^{\otimes}\right]
    \\
    &=
    \frac{1}{(p_{\cS}^{\ast})^2}
    \sum_{k = 1}^K
    p_{\cS, k}^{\ast}
    (1 - p_{\cS, k}^{\ast})
    \bP (X \in \cX_k)
    \\
    &
    \hspace{6em}
    \cdot
    \left\{
        \begin{pmatrix}
            \bE [e_0^\ast (X; \beta_0^\ast)^{\otimes 2} \mid X \in \cX_k] \\
            \vdots \\
            \bE [e_J^\ast (X; \beta_J^\ast)^{\otimes 2} \mid X \in \cX_k]
        \end{pmatrix}
        -
        \begin{pmatrix}
            \bE [e_0^\ast (X; \beta_0^\ast) \mid X \in \cX_k] \\
            \vdots \\
            \bE [e_J^\ast (X; \beta_J^\ast) \mid X \in \cX_k]
        \end{pmatrix}^{\otimes 2}
    \right\}
    \\
    &=
    \frac{1}{(p_{\cS}^{\ast})^2}
    \sum_{k = 1}^K
    p_{\cS, k}^{\ast}
    (1 - p_{\cS, k}^{\ast})
    \bP (X \in \cX_k)
    \bV \left(
        \begin{pmatrix}
            e_0^\ast (X; \beta_0^\ast) \\
            \vdots \\ 
            e_J^\ast (X; \beta_J^\ast)
        \end{pmatrix} 
        \ \bigg | \
        X \in \cX_k
    \right)
    ,
\end{align*}
and similarly,
\begin{align*}
    \Delta_K^{p \to k}
    &=
    \bE \left[F^p (W)^{\otimes}\right]
    -
    \bE \left[F^k (W)^{\otimes}\right]
    \\
    &=
    \frac{1}{(p_{\cS}^{\ast})^2}
    \sum_{k = 1}^K
    p_{\cS, k}^{\ast}
    (1 - p_{\cS, k}^{\ast})
    \bP (X \in \cX_k)
    \left\{
        \bE \left[
            \begin{pmatrix}
                e_0^\ast (X; \beta_0^\ast) \\
                \vdots \\ 
                e_J^\ast (X; \beta_J^\ast)
            \end{pmatrix} 
            \ \bigg | \
            X \in \cX_k
        \right]
    \right\}^{\otimes 2}
\end{align*}
holds.
\end{proof}

\begin{remark}
As we show in Proposition \ref{prop:projection-of-IF}, $F^{uk} - F^p$ and $F^p$ are orthogonal.
It is easy to verify it in this specific example as follows:
\begin{align*}
    &\hphantom{=}
    \bE \left[
        (F_j^{uk} (W) - F_j^p (W))^\prime
        F_j^p (W)
    \right]
    \\
    &=
    \frac{1}{(p_{\cS}^{\ast})^2}
    \bE \Bigg[
        \sum_{k = 1}^K 
        \bI \{X \in \cX_k\}
        (D_{\cS} - p_{\cS, k}^{\ast})
        \left(
            e_j^\ast (X; \beta_j^\ast)
            -
            \bE [e_j^\ast (X; \beta_j^\ast) \mid X \in \cX_k]
        \right)^\prime
        \\
        &\hspace{8cm}\cdot
        \left\{
            \sum_{k = 1}^K 
            \bI \{X \in \cX_k\}
            (D_{\cS} - p_{\cS, k}^{\ast})
            \bE [e_j^\ast (X; \beta_j^\ast) \mid X \in \cX_k]
        \right\}
    \Bigg]
    \\
    &=
    \frac{1}{(p_{\cS}^{\ast})^2}
    \sum_{k = 1}^K 
    \bE \left[
        \bI \{X \in \cX_k\}
        (D_{\cS} - p_{\cS, k}^{\ast})^2
        \left(
            e_j^\ast (X; \beta_j^\ast)
            -
            \bE [e_j^\ast (X; \beta_j^\ast) \mid X \in \cX_k]
        \right)^\prime
    \right]
    \bE [e_j^\ast (X; \beta_j^\ast) \mid X \in \cX_k]
    \\
    &=
    \frac{1}{(p_{\cS}^{\ast})^2}
    \sum_{k = 1}^K 
    p_{\cS, k}^{\ast} (1 - p_{\cS, k}^{\ast})
    \bE \left[
        \bI \{X \in \cX_k\}
        \left(
            e_j^\ast (X; \beta_j^\ast)
            -
            \bE [e_j^\ast (X; \beta_j^\ast) \mid X \in \cX_k]
        \right)^\prime
    \right]
    \bE [e_j^\ast (X; \beta_j^\ast) \mid X \in \cX_k]
    \\
    &=
    0
    .
\end{align*}
\end{remark}

\subsection{Proof of Proposition \ref{prop:limit-stratified-propsnsoty-score-fine}}
\begin{proof}
Notice that
\begin{align*}
    \Delta_{K (n)}^{uk \to p}
    =
    \left(
        \bE \left[F^{uk} (W)^{\otimes}\right]
        -
        \bE \left[F^k (W)^{\otimes}\right]
    \right)
    -
    \left(
        \bE \left[F^p (W)^{\otimes}\right]
        -
        \bE \left[F^k (W)^{\otimes}\right]
    \right)
    .
\end{align*}
For the first parenthesis, we have
\begin{align*}
    \bE \left[F^{uk} (W)^{\otimes}\right]
    -
    \bE \left[F^k (W)^{\otimes}\right]
    &=
    \frac{1}{(p_{\cS}^{\ast})^2}
    \sum_{k = 1}^{K (n)}
    p_{\cS, k}^{\ast}
    (1 - p_{\cS, k}^{\ast})
    \bP (X \in \cX_k)
    \bE \left[
        \begin{pmatrix}
            e_0^\ast (X; \beta_0^\ast) \\
            \vdots \\ 
            e_J^\ast (X; \beta_J^\ast)
        \end{pmatrix}^{\otimes 2} 
        \ \Bigg | \
        X \in \cX_k
    \right]
    \\
    &\overset{K \to \infty}{\longrightarrow}
    \frac{1}{\left(p_{\cS}^{\ast}\right)^2}
    \int_{\cX}
        p_{\cS}^{\ast} (x)
        (1 - p_{\cS}^{\ast} (x))
        \begin{pmatrix}
            e_0^\ast (x; \beta_0^\ast) \\
            \vdots \\ 
            e_J^\ast (x; \beta_J^\ast)
        \end{pmatrix}^{\otimes 2} 
    d F_X (x)
    ,
\end{align*}
and similarly, for the second parenthesis, it holds that 
\begin{align*}
    \bE \left[F^p (W)^{\otimes}\right]
    -
    \bE \left[F^k (W)^{\otimes}\right]
    &=
    \frac{1}{\left(p_{\cS}^{\ast}\right)^2}
    \sum_{k = 1}^{K (n)}
    p_{\cS, k}^{\ast}
    (1 - p_{\cS, k}^{\ast})
    \bP (X \in \cX_k)
    \bE \left[
        \begin{pmatrix}
            e_0^\ast (X; \beta_0^\ast) \\
            \vdots \\ 
            e_J^\ast (X; \beta_J^\ast) \\
        \end{pmatrix} 
        \ \Bigg | \
        X \in \cX_k
    \right]^{\otimes 2}
    \\
    &\overset{K \to \infty}{\longrightarrow}
    \frac{1}{\left(p_{\cS}^{\ast}\right)^2}
    \int_{\cX}
        p_{\cS}^{\ast} (x)
        (1 - p_{\cS}^{\ast} (x))
        \begin{pmatrix}
            e_0^\ast (x; \beta_0^\ast) \\
            \vdots \\ 
            e_J^\ast (x; \beta_J^\ast)
        \end{pmatrix}^{\otimes 2}
    d F_X (x)
    .
\end{align*}
Therefore, it holds that $\lim_{n \to \infty} \Delta_{K (n)}^{uk \to p} = 0.$
\end{proof}

\subsection{Proof of Proposition \ref{prop:limit-stratified-propsnsoty-score-coarse}} \label{app:proof-limit-stratified-propsnsoty-score-coarse}

\begin{proof}
There exists a sequence $k_n$ such that $\cX_{k_n}^{(n)}$ is nonincreasing and $\lim_{n \to \infty} \mathrm{vol} (\cX_{k_n}^{(n)}) > 0.$
Thus, there is an open ball $B$ such that $\cap_n \cX_{k_n}^{(n)} \supset B.$
It holds that
\begin{align*}
    \Delta_{K (n)}^{uk \to p}
    \geq
    \frac{1}{(p_{\cS}^{\ast})^2}
    p_{\cS, k_n}^\ast
    (1 - p_{\cS, k_n}^\ast)
    \bP (X \in \cX_{k_n})
    \bV \left(
        \begin{pmatrix}
            e_0^\ast (X; \beta_0^\ast) \\
            \vdots \\ 
            e_J^\ast (X; \beta_J^\ast)
        \end{pmatrix} 
        \ \bigg | \
        X \in B
    \right)
    >
    0
\end{align*}
because $e_j^\ast (\cdot; \beta_j^\ast)$ is not constant on $B.$
\end{proof}

\subsection{Convergence of a bounded and nondecreasing sequence of positive semidefinite matrices} \label{app:conv-psd}

We show the following statement.
\begin{lemma}
Let $A_n \in \bR^{d \times d}$ be a nondecreasing sequence of positive semidefinite matrices that is bounded above by another positive semidefinite matrix.
Then, $\lim_{n \to \infty} A_n$ exists.
\end{lemma}

\begin{proof}
For each $x \in \bR^d,$ a real sequence $(x^\prime A_n x)_{n \in \bN}$ is nondecreasing and bounded, so that its limit exists.
The conclusion holds because the $(i, j)$-element of $A_n$ can be written as
\begin{align*}
    A_{n, (i, j)}
    =
    \frac{1}{4}
    \left(
        (e_i + e_j)^\prime A_n (e_i + e_j)
        -
        (e_i - e_j)^\prime A_n (e_i - e_j)
    \right)
    ,
\end{align*}
which converges.
\end{proof}

\subsection{Proof of Theorem \ref{thm:many-param-limit}} \label{app:proof-many-param-limit}

Let
\begin{align*}
    \cH_2^{p, n}
    &\coloneqq
    \left\{
        w
        \mapsto
        c_n^\prime
        \sum_{i \in \cT}
        d_i
        S_i^{\ast n} (x)
        \ \Big | \ 
        c_n \in \bR^{d (n)}
    \right\}
    , 
    \\
    \cH_2^{uk}
    &\coloneqq
    \left\{
        w
        \mapsto
        \sum_{i \in \cT}
        d_i
        \frac{\dot p_i (x)}{p_i^\ast (x)}
        \ \Big | \ 
        \dot p_i \in L^2 (F_X)
        ,
        \ 
        \sum_{i \in \cT} \dot p_i 
        =
        0
    \right\}
    .
\end{align*}
Also, let $h_j^{p, n} \coloneqq (h_{j, 1}^{p, n}, \dots, h_{j, d_m}^{p, n})^\prime$ and $h_j^{uk} \coloneqq (h_{j, 1}^{uk}, \dots, h_{j, d_m}^{nk})^\prime$ where 
\begin{align*}
    h_{j, \ell}^{p, n} (W)
    &\coloneqq
    \frac{1}{p_{\cS}^{\ast}} 
    \left(c_{j, \ell}^n(\beta_i^\ast, e_j^\ast, \gamma^\ast)\right)^\prime
    \sum_{i \in \cT} D_i S_i^{\ast n}(X)
    \in
    \cH_2^{p, n}
    \\
    h_{j, \ell}^{uk} (W)
    &\coloneqq
    \frac{1}{p_{\cS}^{\ast}}
    (D_{\cS} - p_{\cS}^{\ast} (X)) 
    e_{j, \ell}^{\ast} (X; \beta_j^\ast)
    \in 
    \cH_2^{uk}
\end{align*}
for
\begin{align*}
    c_{j, \ell}^n(\beta_j^\ast, e_j^\ast, \gamma^\ast)
    =
    \bE \left[
        \left(\sum_{i \in \cT} D_i S_i^{\ast n} (X)\right)^{\otimes 2}
    \right]^{-1}
    \bE \left[
        e_{j, \ell}^{\ast} (X; \beta_j^\ast)
        \sum_{i \in \cS}
        D_i
        S_i^{\ast n} (X)
    \right]
    .
\end{align*}
The fact that $h_{j, \ell}^{uk} \in \cH_2^{uk}$ follows by setting
\begin{align*}
    \dot p_i (x)
    =
    \begin{cases}
        \frac{1}{p_{\cS}^{\ast}}
        e_{j, \ell}^{\ast} (X; \beta_j^\ast)
        (1 - p_{\cS}^{\ast} (X))
        p_i^\ast (X)
        &
        (i \in \cS)
        \\
        -
        \frac{1}{p_{\cS}^{\ast}}
        e_{j, \ell}^{\ast} (X; \beta_j^\ast)
        p_{\cS}^{\ast} (X)
        p_i^\ast (X)
        &
        (i \notin \cS)
    \end{cases}
    .
\end{align*}
Let $h^{p, n} \coloneqq ((h_0^{p, n})^\prime, \dots, (h_J^{p, n})^\prime)^\prime \in (\cH_2^{p, n})^{(J + 1) d_m}$ and $h^{uk} \coloneqq ((h_0^{uk})^\prime, \dots, (h_J^{uk})^\prime)^\prime \in (\cH_2^{uk})^{(J + 1) d_m}.$
Notice that $h^{p, n}$ is the projection of $h^{uk}$ onto $(\cH_2^{p, n})^{(J + 1) d_m}$ by Proposition \ref{prop:projection-of-IF}.

We show the following theorem, which immediately implies Theorem \ref{thm:many-param-limit}.

\begin{theorem}
The following are equivalent:
\begin{enumerate}
    \item[(1)] Condition \ref{cond:score-is-flexible} holds;
    \item[(2)] $h^{p, n} \to h^{uk}$ in $(L^2(F_{D, X}))^{(J + 1) d_m};$
    \item[(3)] $V^{p, n} \nearrow V^{uk}$ in the sense of positive semidefinite matrix.
\end{enumerate}
\end{theorem}

\begin{proof}
$(2 \Rightarrow 1)$
Let $c_{j, \ell}^{\ast n} \coloneqq c_{j, \ell}^n(\beta_j^\ast, e_j^\ast, \gamma^\ast)$ for $j \in \cT$ and $\ell \in \{1, \dots, d_m\}.$
Then, it holds that
\begin{align*}
    \bE \left[
        \left|
            h_{j, \ell}^{p, n} (W)
            -
            h_{j, \ell}^{uk} (W)
        \right|^2
    \right]
    &=
    \frac{1}{(p_{\cS}^{\ast})^2}
    \bE \left[
        \sum_{i \in \cT}
        p_i^\ast (X)
        \left\{
            (c_{j, \ell}^{\ast n})^\prime
            S_i^{\ast n} (X)
            -
            (\bI\{i \in \cS\} - p_{\cS}^{\ast} (X))
            e_{j, \ell}^{\ast} (X; \beta_j^\ast)
        \right\}^2
    \right]
    \\
    &\geq
    \frac{p_{\text{min}}}{(p_{\cS}^{\ast})^2}
    \sum_{i \in \cT}
    \bE \left[
        \left\{
            (c_{j, \ell}^{\ast n})^\prime
            S_i^{\ast n} (X)
            -
            (\bI\{i \in \cS\} - p_{\cS}^{\ast} (X))
            e_{j, \ell}^{\ast} (X; \beta_j^\ast)
        \right\}^2
    \right]
    .
\end{align*}
Since the LHS converges to zero as $n \to \infty$ by (2), we have
\begin{align*}
    \inf_{c_n \in \bR^{d (n)}}
    \norm{
        c_n^\prime \bS_n
        -
        \bar \ve_{j, \ell}
    }_{(L^2 (F_X))^J}
    \leq
    \sum_{i = 1}^J
    \norm{(c_{j, \ell}^{\ast n})^\prime S_i^{\ast n} - (\bI\{i \in \cS\} - p_{\cS}^{\ast} (\cdot)) e_{j, \ell}^{\ast} (\cdot; \beta_j^\ast)}_{L^2 (F_X)}
    \to
    0
    .
\end{align*}

\noindent
$(1 \Rightarrow 2)$
Let $c_n \in \bR^{d (n)}.$
Since $h_{j, \ell}^{p, n}$ is the projection of $h_{j, \ell}^{uk}$ onto $\cH_2^{p, n},$ it holds
\begin{align*}
    \bE \left[
        \left|
            h_{j, \ell}^{p, n} (W)
            -
            h_{j, \ell}^{uk} (W)
        \right|^2
    \right]
    &\leq
    \frac{1}{(p_{\cS}^{\ast})^2}
    \bE \left[
        \left\{
            c_n^\prime 
            \sum_{i \in \cT} D_i S_i^{\ast n} (X)
            -
            (D_{\cS} - p_{\cS}^{\ast} (X)) e_{j, \ell}^{\ast} (X; \beta_j^\ast)
        \right\}^2
    \right]
    \\
    &=
    \frac{1}{(p_{\cS}^{\ast})^2}
    \bE \left[
        \left|
            \sum_{i \in \cT} 
            D_i 
            \left\{
                c_n^\prime 
                S_i^{\ast n} (X)
                -
                (\bI\{i \in \cS\} - p_{\cS}^{\ast} (X)) e_{j, \ell}^{\ast} (X; \beta_j^\ast)
            \right\}
        \right|^2
    \right]
    \\
    &\leq
    \frac{2 J}{(p_{\cS}^{\ast})^2}
    \bE \left[
        \sum_{i = 1}^J
        \left\{
            c_n^\prime
            S_i^{\ast n} (X)
            -
            (\bI\{i \in \cS\} - p_{\cS}^{\ast} (X))
            e_{j, \ell}^{\ast} (X; \beta_j^\ast)
        \right\}^2
    \right]
    \\
    &
    \hspace{1em}
    +
    \frac{2}{(p_{\cS}^{\ast})^2}
    \bE \left[
        \left\{
            -
            c_n^\prime
            \sum_{i = 1}^J
            S_i^{\ast n} (X)
            \frac{p_i^\ast (X)}{p_0^\ast (X)}
            -
            (\bI\{0 \in \cS\} - p_{\cS}^{\ast} (X))
            e_{j, \ell}^{\ast} (X; \beta_j^\ast)
        \right\}^2
    \right]
    .
\end{align*}
For the second term, since
\begin{align*}
    \bI\{0 \in \cS\} 
    - 
    p_{\cS}^{\ast} (X)
    =
    -
    \sum_{i = 1}^J
    (\bI\{i \in \cS\} - p_{\cS}^{\ast} (X))
    \frac{p_i^\ast (X)}{p_0^\ast (X)}
    ,
\end{align*}
it holds that
\begin{align*}
    &\hphantom{=}
    \frac{2}{(p_{\cS}^{\ast})^2}
    \bE \left[
        \left\{
            -
            c_n^\prime
            \sum_{i = 1}^J
            S_i^{\ast n} (X)
            \frac{p_i^\ast (X)}{p_0^\ast (X)}
            -
            (\bI\{0 \in \cS\} - p_{\cS}^{\ast} (X))
            e_{j, \ell}^{\ast} (X; \beta_j^\ast)
        \right\}^2
    \right]
    \\
    &=
    \frac{2}{(p_{\cS}^{\ast})^2}
    \bE \left[
        \left|
            \sum_{i = 1}^J
            \left\{
                c_n^\prime
                S_i^{\ast n} (X)
                -
                (\bI\{i \in \cS\} - p_{\cS}^{\ast} (X))
                e_{j, \ell}^{\ast} (X; \beta_j^\ast)
            \right\}
            \frac{p_i^\ast (X)}{p_0^\ast (X)}
        \right|^2
    \right]
    \\
    &\leq
    \frac{2 J}{(p_{\cS}^{\ast})^2 p_{\text{min}}^2}
    \sum_{i = 1}^J
    \bE \left[
        \left\{
            c_n^\prime
            S_i^{\ast n} (X)
            -
            (\bI\{i \in \cS\} - p_{\cS}^{\ast} (X))
            e_{j, \ell}^{\ast} (X; \beta_j^\ast)
        \right\}^2
    \right]
    .
\end{align*}
Therefore, we have
\begin{align*}
    \bE \left[
        \left|
            h_{j, \ell}^{p, n} (W)
            -
            h_{j, \ell}^{uk} (W)
        \right|^2
    \right]
    \leq
    \frac{2 J (p_{\text{min}}^2 + 1)}{(p_{\cS}^{\ast})^2 p_{\text{min}}^2}
    \sum_{i = 1}^J
    \bE \left[
        \left\{
            c_n^\prime
            S_i^{\ast n} (X)
            -
            (\bI\{i \in \cS\} - p_{\cS}^{\ast} (X))
            e_{j, \ell}^{\ast} (X; \beta_j^\ast)
        \right\}^2
    \right]
    .
\end{align*}
By taking the infimum over $c_n \in \bR^{d (n)}$ in the RHS, Condition \ref{cond:score-is-flexible} implies that
\begin{align*}
    \bE \left[
        \left|
            h_{j, \ell}^{p, n} (W)
            -
            h_{j, \ell}^{uk} (W)
        \right|^2
    \right]
    \to
    0
    .
\end{align*}

\noindent
$(2 \Rightarrow 3).$
It is sufficient to show $\norm{\bE \left[h^{uk} (W)^{\otimes 2} \right] - \bE \left[h^{p, n} (W)^{\otimes 2} \right]}_{F} \to 0,$ where $\norm{\cdot}_F$ denotes the Frobenius norm.
It holds that 
\begin{align*}
    \norm{
        \bE \left[
            h^{uk} (W)^{\otimes 2} 
        \right]
        -
        \bE \left[
            h^{p, n} (W)^{\otimes 2} 
        \right]
    }_{F}^2
    &=
    \sum_{i, j, \ell, k}
    \left(
        \bE \left[
            h_{i, \ell}^{uk} (W)
            h_{j, k}^{uk} (W)
            -
            h_{i, \ell}^{p, n} (W)
            h_{j, k}^{p, n} (W)
        \right]
    \right)^2
    \\
    &\leq
    \sum_{i, j, \ell, k}
    \norm{
        h_{i, \ell}^{uk}
        h_{j, k}^{uk}
        -
        h_{i, \ell}^{p, n}
        h_{j, k}^{p, n}
    }_{L^1 (F_{D, X})}^2
    .
\end{align*}
For each $i, j, \ell, k,$ we have by (2),
\begin{align*}
    \norm{
        h_{i, \ell}^{uk}
        h_{j, k}^{uk}
        -
        h_{i, \ell}^{p, n}
        h_{j, k}^{p, n}
    }_{L^1 (F_{D, X})}
    &\leq
    \norm{
        h_{i, \ell}^{uk}
        \left(
            h_{j, k}^{uk}
            -
            h_{j, k}^{p, n}
        \right)
    }_{L^1 (F_{D, X})}
    +
    \norm{
        \left(
            h_{i, \ell}^{uk}
            -
            h_{i, \ell}^{p, n}
        \right)
        h_{j, k}^{p, n}
    }_{L^1 (F_{D, X})}
    \\
    &\leq
    \norm{
        h_{i, \ell}^{uk}
    }_{L^2 (F_{D, X})}
    \norm{
        h_{j, k}^{uk}
        -
        h_{j, k}^{p, n}
    }_{L^2 (F_{D, X})}
    +
    \norm{
        h_{j, k}^{p, n}
    }_{L^2 (F_{D, X})}
    \norm{
        h_{i, \ell}^{uk}
        -
        h_{i, \ell}^{p, n}
    }_{L^2 (F_{D, X})}
    \\
    &\to
    0
    ,
\end{align*}
which implies (3).

\noindent
$(3 \Rightarrow 2).$
Since $h^{p, n}$ is the projection of $h^{uk}$ onto $(\cH_2^{p, n})^{(J + 1) d_m},$ it holds that
\begin{align*}
    \bE \left[\norm{h^{uk} (W)}^2 \right]
    =
    \bE \left[\norm{h^{p, n} (W)}^2 \right]
    +
    \bE \left[\norm{h^{uk} (W) - h^{p, n} (W)}^2 \right]
    .
\end{align*}
By $(3),$ we have
\begin{align*}
    \bE \left[
        \norm{h^{p, n} (W)}^2 
    \right]
    =
    \bE \left[
        \norm{h^{p, n} (W)^{\otimes 2}}_F
    \right]
    \to
    \bE \left[
        \norm{h^{uk} (W)^{\otimes 2}}_F
    \right]
    =
    \bE \left[
        \norm{h^{uk} (W)}^2 
    \right]
    ,
\end{align*}
which implies $(2).$
\end{proof}

\begin{proof} [of Theorem \ref{thm:many-param-limit}]
It follows from the equivalence between conditions (1) and (3).
\end{proof}

\subsection{Proof of Proposition \ref{prop:full-logit}}
\begin{proof}
Since it holds
\begin{align*}
    \frac{\partial}{\partial \Gamma_{k, j}} \at{\gamma = \gamma^\ast} \log p_j (x; \gamma)
    =
    b_k (x) (1 - p_j^\ast (x))
    ,
\end{align*}
the score of this parametric model is 
\begin{align*}
    S_j^{\ast n} (x)
    &=
    \frac{\partial}{\partial \gamma^{(n)}}
    \at{\gamma = \gamma^\ast}
    \log p_{j} (x; \gamma)
    \\
    &=
    \begin{pmatrix}
        \frac{\partial}{\partial \Gamma_{1, j}} \at{\gamma = \gamma^\ast} \log p_j (x; \gamma) \\
        \vdots \\
        \frac{\partial}{\partial \Gamma_{n, j}} \at{\gamma = \gamma^\ast} \log p_j (x; \gamma)
    \end{pmatrix}
    \otimes
    e_j^J
    \\
    &=
    (1 - p_j^\ast (x))
    \begin{pmatrix}
        b_1 (x) \\
        \vdots \\
        b_n (x)
    \end{pmatrix}
    \otimes
    e_j^J
\end{align*}
for $1 \leq j \leq J,$ and therefore,
\begin{align*}
    \bS_n (x)
    =
    \begin{pmatrix}
        S_1^{\ast n} (x) & \dots & S_J^{\ast n} (x)
    \end{pmatrix}
    =
    \begin{pmatrix}
        b_1 (x) \\
        \vdots \\
        b_n (x)
    \end{pmatrix}
    \otimes
    \begin{pmatrix}
        1 - p_1^\ast (x) && \\
        &\ddots& \\
        && 1 - p_J^\ast (x)
    \end{pmatrix}
    .
\end{align*}
For $c_n = (c_{n1}^\prime, \dots, c_{nn}^\prime)^\prime \in \bR^{d (n)},$ where $c_{nk} \in \bR^J,$ it holds that
\begin{align*}
    c_n^\prime \bS_n (x)
    =
    \begin{pmatrix}
        (1 - p_1^\ast (x)) \sum_{k = 1}^n c_{nk1} b_k (x) & \dots & (1 - p_J^\ast (x)) \sum_{k = 1}^n c_{nkJ} b_k (x)
    \end{pmatrix}
    .
\end{align*}
Thus, the $i$th element of $c_n^\prime \bS_n (x) - \bar \ve_{j, \ell} (x)$ is
\begin{align*}
    (1 - p_i^\ast (x))
    \left(
        \sum_{k = 1}^n c_{nki} b_k (x)
        -
        \frac{\bI\{i \in \cS\} - p_{\cS}^{\ast} (x)}{1 - p_i^\ast (x)}
        e_{j, \ell}^{\ast} (x; \beta_j^\ast)
    \right)
    ,
\end{align*}
which implies that the condition in the statement is equivalent to Condition \ref{cond:score-is-flexible}.
We obtain the conclusion by Theorem \ref{thm:many-param-limit}.
\end{proof}

\subsection{Proof of Proposition \ref{prop:deg-logit}} \label{app:proof-deg-logit}

\begin{proof}
Let
\begin{align*}
    \bar p^\ast (x)
    \coloneqq
    \frac{
        \exp(\sum_{k = 1}^\infty \tilde \gamma_k^\ast b_k (x))
    }{1 + J \exp(\sum_{k = 1}^\infty \tilde \gamma_k^\ast b_k (x))}
    .
\end{align*}
In a similar manner to the proof of Proposition \ref{prop:full-logit}, it can be shown that
\begin{align*}
    \bS_n (x)
    =
    \begin{pmatrix}
        b_1 (x) (1 - J \bar p^\ast (x)) & \dots & b_1 (x) (1 - J \bar p^\ast (x)) \\
        &\vdots& \\
        b_J (x) (1 - J \bar p^\ast (x)) & \dots & b_J (x) (1 - J \bar p^\ast (x)) \\
    \end{pmatrix}
    ,
\end{align*}
and for $c_n \in \bR^{d (n)},$
\begin{align*}
    c_n^\prime \bS_n (x)
    =
    \begin{pmatrix}
        (1 - J \bar p^\ast (x)) \sum_{k = 1}^n c_{nk} b_k (x) & \dots & (1 - J \bar p^\ast (x)) \sum_{k = 1}^n c_{nk} b_k (x)
    \end{pmatrix}
    .
\end{align*}
As the rank of this vector is at most one, it cannot approximate $\bar \ve_{j, \ell}$ unless it contains only one function, which is not the case if $\cS \cap \{1, \dots, J\} \neq \emptyset$ and $\{1, \dots, J\} \setminus \cS \neq \emptyset.$
Hence, we have the conclusion.
\end{proof}

\begin{remark}
When the treatment is binary, i.e., $J = 1,$ $\bar \ve_{j, \ell}$ contains only one function.
Thus, if the dictionary $b_n$ is rich enough, $c_n^\prime \bS_n$ can approximate $\bar \ve_{j, \ell},$ which implies that Condition \ref{cond:score-is-flexible} holds and $V^{p, \infty} = V^{uk}.$
\end{remark}

\section{Applying Theorem \ref{thm:many-param-limit} to Stratified Propensity Scores} \label{app:application-to-stratified-experiments}

In this section, we apply the theory developed in Section \ref{sec:param-vs-nonparam} to stratified propensity scores.
In Section \ref{sec:stratified-experiment}, we have defined them as 
\begin{align*}
    p_j (x)
    =
    \sum_{k = 1}^K
    p_{j, k}
    \bI\{x \in \cX_k\}
\end{align*}
where $p_{j, k}$ are parameters for a fixed partition $\{\cX_k\}_{k = 1}^K.$
Since this parameterization depends on the partition and does not fit the framework in Section \ref{sec:param-vs-nonparam}, we need to redefine stratified propensity scores.

First, we construct a tree-diagram-like index.
Let $\bar \cI = \{0, 1, \dots, \bar I\},$ which is an index set at level $1.$
For each $i_1 \in \bar \cI,$ let $\cI_{i_1} = \{0, 1, \dots, I_{i_1}\}$ be an index set at level $2$ that follows $i_1.$
Again, for each $i_2 \in \cI_{i_1},$ let $\cI_{i_1 i_2} = \{0, 1, \dots, I_{i_1 i_2}\}$ be an index set at level $3$ that follows $i_1i_2,$ and define similar index sets sequentially.
Let  $\cI (n) = \{\vi(n) = i_1 \dots i_n \mid i_1 \in \bar \cI, \ i_{k + 1} \in \cI_{i_1 \dots i_k} \ \text{for} \ 1 \leq k \leq n - 1\}$ be the set of possible indices with size $n.$ 
For example, if $\bar \cI = \{0, 1, 2\}, \ \cI_0 = \{0, 1, 2\}, \ \cI_1 = \{0, 1\}, \ \cI_2 = \{0, 1\}, \dots,$ then $\cI(1) = \{0, 1, 2\}$ and $\cI(2) = \{00, 01, 02, 10, 11, 20, 21\}.$
Finally, let $\cI = \cup_{n = 1}^\infty \cI(n)$ be the set of all possible indices.

Consider a family of measurable subsets $\{\cX_{\vi}\}_{\vi \in \cI}$ of the covariate space $\cX$ such that
\begin{align*}
    \cX_{i_1 \dots i_{n - 1}}
    =
    \bigsqcup_{i_n \in \cI_{i_1 \dots i_{n - 1}}} \cX_{i_1 \dots i_{n - 1} i_n}
    \ \text{for} \ n \geq 2,
    \ \text{and} \ 
    \cX 
    = 
    \bigsqcup_{i_1 \in \bar \cI} \cX_{i_1}
    .
\end{align*}
Notice that $\{\cX_{\vi (n)}\}_{\vi(n) \in \cI(n)}$ gives a partition of $\cX$ for each $n,$ and that it is nested in the sense that if $m \geq n,$ then for any $\vi(m) \in \cI(m),$ there uniquely exists $\vi(n) \in \cI(n)$ such that $\cX_{\vi(m)} \subset \cX_{\vi(n)}.$
Also, the first $m$ numbers of such $\vi(n)$ are exactly $\vi(m).$
In particular, $\{\cX_{\vi (n)}\}_{\vi (n) \in \cI (n)}$ is a partition of $\cX$ at level $n.$
\begin{example} \label{eg:partition}
Figure \ref{fig:partition} describes such partitions up to level $2$ for $\bar \cI = \{0, 1, 2\}, \ \cI_0 = \{0, 1, 2\}, \ \cI_1 = \{0, 1\}, \ \cI_2 = \{0, 1\}, \dots.$
As in the figure, it holds that
\begin{align*}
    \cX
    =
    \cX_0 \sqcup \cX_1 \sqcup \cX_2
    , \ 
    \cX_0 
    =
    \cX_{00} \sqcup \cX_{01} \sqcup \cX_{02}
    , \ 
    \cX_1
    =
    \cX_{10} \sqcup \cX_{11}
    , \ \text{and} \ 
    \cX_2
    =
    \cX_{20} \sqcup \cX_{21}
    .
\end{align*}
\end{example}

\begin{figure}
    \centering
    \begin{subfigure}[b]{0.3\textwidth}
        \centering
        \begin{tikzpicture}[scale=0.8]
            \draw[line width=1pt, fill=black, fill opacity=0.1] (0,0) rectangle (5,5);
            
            \node at (0.5,4.5) {$\cX$};
        \end{tikzpicture}
        \caption{Level $0: \cX$}
    \end{subfigure}
    \begin{subfigure}[b]{0.3\textwidth}
        \centering
        \begin{tikzpicture}[scale=0.8]
            \draw[line width=1pt] (0,0) rectangle (5,5);
            \draw [line width=0pt, fill=black, fill opacity=0.1] (0,0) -- (5,0) -- (2.5,2.5) -- (0,2.5) -- cycle;
            \draw [line width=0pt, fill=black, fill opacity=0.1] (5,0) -- (5,5) -- (2.5,5) -- (2.5,2.5) -- cycle;

            \draw[line width=1pt] (0,2.5) to (2.5,2.5);
            \draw[line width=1pt] (2.5,5) to (2.5,2.5);
            \draw[line width=1pt] (5,0) to (2.5,2.5);

            \node at (0.5,4.5) {$\cX_0$};
            \node at (0.5,2) {$\cX_1$};
            \node at (3,4.5) {$\cX_2$};
        \end{tikzpicture}
        \caption{Level $1:  \{\cX_{\vi(1)}\}_{\vi(1) \in \cI(1)}$}
    \end{subfigure}
    \begin{subfigure}[b]{0.3\textwidth}
        \centering
        \begin{tikzpicture}[scale=0.8]
            \draw[line width=1pt] (0,0) rectangle (5,5);
            \draw [line width=0pt, fill=black, fill opacity=0.1] (0,5) -- (1.25,3.75) -- (2.5,3.75) -- (2.5,5) -- cycle;
            \draw [line width=0pt, fill=black, fill opacity=0.1] (1.25,2.5) -- (2.5,2.5) -- (2.5,3.75) -- (1.25,3.75) -- cycle;
            \draw [line width=0pt, fill=black, fill opacity=0.1] (0,0) -- (5,0) -- (3.75,1.25) -- (0,1.25) -- cycle;
            \draw [line width=0pt, fill=black, fill opacity=0.1] (5,0) -- (5,5) -- (3.75,5) -- (3.75,1.25) -- cycle;
            
            \draw[line width=1pt] (0,2.5) to (2.5,2.5);
            \draw[line width=1pt] (2.5,5) to (2.5,2.5);
            \draw[line width=1pt] (5,0) to (2.5,2.5);
            \draw[line width=1pt] (0,5) to (1.25,3.75);
            \draw[line width=1pt] (1.25,2.5) to (1.25,3.75);
            \draw[line width=1pt] (1.25,3.75) to (2.5,3.75);
            \draw[line width=1pt] (0,1.25) to (3.75,1.25);
            \draw[line width=1pt] (3.75,5) to (3.75,1.25);

            \node at (0.5,4) {$\cX_{00}$};
            \node at (1.2,4.5) {$\cX_{01}$};
            \node at (1.7,3.3) {$\cX_{02}$};
            \node at (0.5,2) {$\cX_{10}$};
            \node at (0.5,0.8) {$\cX_{11}$};
            \node at (3,4.5) {$\cX_{20}$};
            \node at (4.3,4.5) {$\cX_{21}$};
        \end{tikzpicture}
        \caption{Level $2: \{\cX_{\vi(2)}\}_{\vi(2) \in \cI(2)}$}
    \end{subfigure}
    \caption{An example of nested partitions of $\cX.$ Cells associated with indices in $\tilde \cI (n)$ are shaded.}
    \label{fig:partition}
\end{figure}
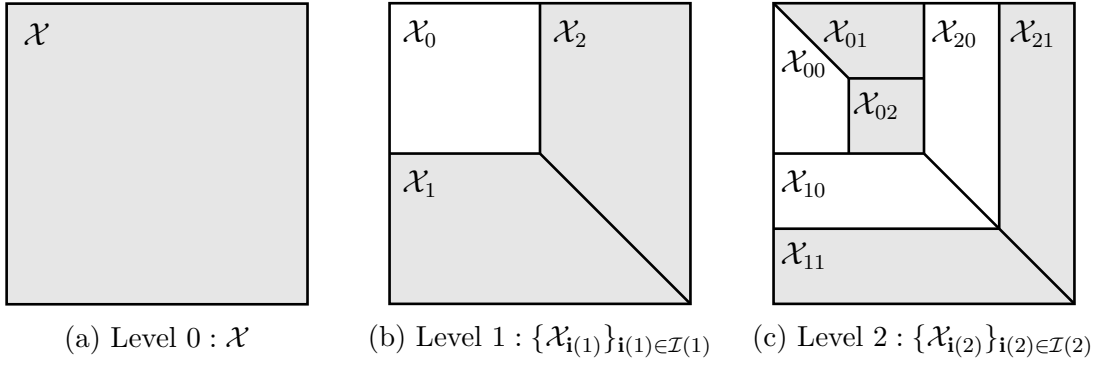

Let $\tilde \cI(n) = \{\vi(n) \in \cI(n) \mid i_n \neq 0\}$ and $\tilde \cI = \cup_{n = 1}^\infty \tilde \cI (n).$
We (re)define a stratified propensity score parameterized by $\gamma = \left(\bar \gamma^\prime, (\gamma_{\vi(1)}^\prime)_{\vi(1) \in \tilde \cI(1)}^\prime, (\gamma_{\vi(2)}^\prime)_{\vi(2) \in \tilde \cI(2)}^\prime, \dots \right)^\prime,$ where $\bar \gamma = (\bar \gamma^1, \dots, \bar \gamma^J)^\prime \in \bR^J$ and $\gamma_{\vi} = (\gamma_{\vi}^1, \dots, \gamma_{\vi}^J)^\prime \in \bR^J$ for each $\vi \in \tilde \cI,$ as
\begin{align} \label{eq:strat-prop-score}
    p_j (x; \gamma)
    =
    \bar \gamma^j
    +
    \sum_{n = 1}^\infty
    \sum_{\vi (n) \in \tilde \cI (n)}
    \gamma_{\vi (n)}^j
    \bI \{x \in \cX_{\vi (n)}\}
\end{align}
for $1 \leq j \leq J,$ and $p_0 (x; \gamma) = 1 - \sum_{j = 1}^J p_j (x; \gamma).$
The base model here is $\cM_\infty^p = \{(p_j (\cdot; \gamma))_{j \in \cT}\}.$
The induced nested parametric models are
\begin{align*}
    \cM_n^p
    =
    \left\{
        (p_j(\cdot; \gamma))_{j \in \cT}
        \mid
        k \geq n + 1
        \Rightarrow
        \gamma^j_{\vi(k)} = 0
        \ \text{for} \ 1 \leq j \leq J
    \right\}
    ,
\end{align*}
and $d (n) = \left(1 + \sum_{k = 1}^n |\tilde \cI(k)|\right)J.$

Let us see why this parameterization makes sense to define a stratified propensity score.
If $k \geq n + 1 \Rightarrow \gamma^j_{\vi(k)} = 0$ is satisfied, the propensity score (\ref{eq:strat-prop-score}) can be written as
\begin{align} \label{eq:strat}
    p_j (x; \gamma)
    =
    \sum_{\vi(n) \in \cI(n)}
    p_{\vi(n)}^j
    \bI \{x \in \cX_{\vi(n)}\}
\end{align}
where
\begin{align} \label{eq:param-strat-prop-score}
    p_{\vi(n)}^j
    =
    \bar \gamma^j
    +
    \sum_{k = 1}^n
    \sum_{\vi(k) \in \tilde \cI(k)}
    \gamma_{\vi(k)}^j
    \bI\{\cX_{\vi(n)} \subset \cX_{\vi(k)}\}
    ,
\end{align}
which implies that $p_j (\cdot; \gamma)$ is stratified in the sense of Section \ref{sec:stratified-experiment}.
Conversely, suppose a propensity score satisfies (\ref{eq:strat}) with parameters $(p_{\vi(n)}^j)_{\vi(n) \in \cI(n), j \in \cT}$ satisfying $\sum_{j \in \cT} p_{\vi(n)}^j = 1.$
Then, there uniquely exists $\left(\bar \gamma^j, (\gamma_{\vi(1)}^j)_{\vi(1) \in \tilde \cI(1)}, \dots, (\gamma_{\vi(n)}^j)_{\vi(n) \in \tilde \cI(n)}\right)$ such that it satisfies the relationship (\ref{eq:param-strat-prop-score}).
For these parameters, (\ref{eq:strat}) is written as (\ref{eq:strat-prop-score}).
Therefore, the parameterization (\ref{eq:strat-prop-score}) successfully represents the stratified propensity score.

\begin{remark}
One might wonder why we do not use $\left((p_{\vi(1)}^j)_{\vi(1) \in \cI(1)}, \dots, (p_{\vi(n)}^j)_{\vi(n) \in \cI(n)}\right)$ directly, which seems more straightforward than using $\gamma$'s.
Unfortunately, this causes an over-parameterization problem.
That is, for a given stratified propensity score, there are multiple parameters that produce it.
This is why we put parameters on $\{\cX_{\vi}\}_{\vi \in \tilde \cI}$ rather than $\{\cX_{\vi}\}_{\vi \in \cI}.$
\end{remark}

By applying Theorem \ref{thm:many-param-limit} to this setup, we have the following proposition.
\begin{proposition} \label{prop:full-strat}
For a stratified propensity score, $V^{p, \infty} = V^{uk}$ if $\left(1, \left(\bI\{\cdot \in \cX_{\vi}\}\right)_{\vi \in \tilde \cI}\right)$ spans $L^2 (F_X).$
\end{proposition}

\begin{proof}
The score of $\cM_n^p$ is
\begin{align*}
    S_j^{\ast n} (x)
    =
    \frac{\partial}{\partial \gamma^{(n)}}
    \at{\gamma = \gamma^\ast}
    \log p_{j} (x; \gamma)
    =
    \frac{1}{p_j^\ast (x)}
    \begin{pmatrix}
        1 \\
        \left(\bI\{x \in \cX_{\vi(1)}\}\right)_{\vi(1) \in \tilde \cI(1)} \\
        \vdots \\
        \left(\bI\{x \in \cX_{\vi(n)}\}\right)_{\vi(n) \in \tilde \cI(n)}
    \end{pmatrix}
    \otimes
    e_j^J
    ,
\end{align*}
and
\begin{align*}
    \bS_n (x)
    &=
    \begin{pmatrix}
        S_1^{\ast n} (x) & \dots & S_J^{\ast n} (x)
    \end{pmatrix}
    \\
    &=
    \begin{pmatrix}
        1 \\
        \left(\bI\{x \in \cX_{\vi(1)}\}\right)_{\vi(1) \in \tilde \cI(1)} \\
        \vdots \\
        \left(\bI\{x \in \cX_{\vi(n)}\}\right)_{\vi(n) \in \tilde \cI(n)}
    \end{pmatrix}
    \otimes
    \begin{pmatrix}
        1 / p_1^\ast (x) && \\
        &\ddots& \\
        && 1 / p_J^\ast (x)
    \end{pmatrix}
    \in
    \bR^{d (n) \times J}
    .
\end{align*}
Therefore, for $c_n = ((c_{n0j})_{1 \leq j \leq J}^\prime, (c_{n \vi(1) j})_{1 \leq j \leq J, \vi(1) \in \tilde \cI(1)}^\prime, \dots, (c_{n \vi(n) j})_{1 \leq j \leq J, \vi(n) \in \tilde \cI(n)}^\prime)^\prime \in \bR^{d (n)},$
\begin{align*}
    c_n^\prime \bS_n (x)
    =
    \begin{pmatrix}
        \displaystyle
        \frac{c_{n01}}{p_1^\ast (x)}
        + 
        \sum_{k = 1}^n 
        \sum_{\vi(k) \in \tilde \cI(k)} 
        \frac{c_{n \vi(k) 1}}{p_1^\ast (x)} \bI\{x \in \cX_{\vi(k)}\}
        & \dots & 
        \displaystyle
        \frac{c_{n0J}}{p_J^\ast (x)}
        + 
        \sum_{k = 1}^n 
        \sum_{\vi(k) \in \tilde \cI(k)} 
        \frac{c_{n \vi(k) J}}{p_J^\ast (x)} \bI\{x \in \cX_{\vi(k)}\}
    \end{pmatrix}
    .
\end{align*}
If  $\left(1, \left(\bI\{\cdot \in \cX_{\vi}\}\right)_{\vi \in \tilde \cI}\right)$ spans $L^2 (F_X),$ then Condition \ref{cond:score-is-flexible} is satisfied.
Hence, the conclusion follows.
\end{proof}

For $\left(1, \left(\bI\{\cdot \in \cX_{\vi}\}\right)_{\vi \in \cI}\right)$ to span $L^2 (F_X),$ each $\cX_{\vi}$ must get finer as the level of the index gets deeper.
This statement implies that if this holds, knowing how the covariate space is split is not beneficial in terms of efficiency.

The following gives an insightful sufficient condition for a parametric restriction to improve the efficiency even asymptotically.
\begin{proposition} \label{prop:deg-strat}
For all $j \in \cT$ and $\ell \in \{1, \dots, d_m\},$ suppose that $e_{j, \ell}^{\ast} (\cdot; \beta^\ast)$ is not constant on any open ball in $\cX.$
For a stratified propensity score, $V^{p, \infty} < V^{uk}$ if 
\begin{align*}
    \lim_{n \to \infty} \max_{\vi(n) \in \cI(n)} \mathrm{vol} (\cX_{\vi(n)})
    >
    0
    .
\end{align*}
\end{proposition}

\begin{proof}
Let $\delta = \lim_{n \to \infty} \max_{\vi (n) \in \cI (n)} \mathrm{vol} (\cX_{\vi (n)}) > 0.$
Then, there exists a sequence $\vi (n) \in \cI (n)$ such that $\cX_{\vi(n)} \supset \cX_{\vi(n + 1)}$ and $\lim_{n \to \infty} \mathrm{vol} (\cX_{\vi (n)}) = \delta.$
Let $\cX_\delta = \cap_n \cX_{\vi (n)}.$
Since $\mathrm{vol} (\cX_\delta) > 0,$ it contains an open ball, so that $e_{j, \ell}^{\ast} (\cdot; \beta_j^\ast)$ is not constant on $\cX_\delta.$
Hence, $c_n^\prime \bS_n$ cannot approximate $\bar \ve_{j, \ell}$ on $\cX_\delta.$
\end{proof}

The condition on the partition holds if, for example, there is a part where the partition does not become finer than a certain level, or formally, if there exist $n$ and a sequence $(\vi(m) \in \cI(m))_{m \geq n}$ such that $\cX_{\vi(n)} = \cX_{\vi(m)}.$
This is exactly when one assumes that the true propensity score is constant on an unignorable region $\cX_{\vi(n)},$ which consequently restricts the space of possible propensity scores by removing the possibility that it takes different values there.

\section{Numerical Studies}
In this section, we evaluate efficiency gains from knowing the propensity score numerically.
For a binary treatment, we consider semiparametric efficiency bounds of the ATT, $\bE [Y_1 - Y_0 \mid T = 1],$ for different specifications.
The setup, which is partly inspired by \cite{lee2018efficient}, is as follows.
A two-dimensional covariate vector $(X_1, X_2)$ is uniformly distributed on $\cX = (- 0.5, 0.5]^2.$
Potential outcomes are $Y_0 = 3 X_1 - 3 X_2 + U_0$ and $Y_1 = 5 + 5 X_1 + X_2 + U_1,$ where $U_0$ and $U_1$ are drawn from $N (0, 1)$ independently of other variables.
The true propensity score is $p_0^\ast (x) = p_1^\ast (x) = 0.5.$

We assume the propensity score is known to be stratified.
We consider the following two partitions: 
\begin{enumerate}[label=(\arabic*)]
    \item for $n \geq 0,$
    \begin{align*}
        \cI_n^1
        \coloneqq
        \left\{
            \left(-0.5 + \frac{k_1}{2^n}, -0.5 + \frac{k_1 + 1}{2^n}\right] 
            \times
            \left(-0.5 + \frac{k_2}{2^n}, -0.5 + \frac{k_2 + 1}{2^n}\right] 
            \ \bigg | \
            0 \leq k_1, k_2 \leq 2^n - 1
        \right\}
    \end{align*}
    \item $\cI_0^2 = \{\cX\},$ and for $n \geq 1,$
    \begin{align*}
        \cI_n^2
        \coloneqq
        \text{
            \footnotesize
            $
            \left\{
                \left(-0.5 + \frac{k_1}{2^n}, -0.5 + \frac{k_1 + 1}{2^n}\right] 
                \times
                \left(-0.5 + \frac{k_2}{2^n}, -0.5 + \frac{k_2 + 1}{2^n}\right] 
                \ \bigg | \
                0 \leq k_1 \leq 2^{n - 1} - 1
                , \
                0 \leq k_2 \leq 2^n - 1
            \right\}
            $
        }
        \\
        \cup
        \left\{
            \left(0, 0.5\right] 
            \times
            (-0.5, 0.5]
        \right\}
        .
    \end{align*}
\end{enumerate}

Partition (1) equally splits $\cX,$ and each cell gets smaller as $n$ increases.
By Proposition \ref{prop:limit-stratified-propsnsoty-score-fine}, we know that the efficiency gain $V_{\mathrm{ATT}}^{uk} - V_{\mathrm{ATT}}^p$ from knowing the partition vanishes as $n$ goes to the infinity.
In partition (2), on the other hand, the right half of $\cX$ remains coarse.
By Proposition \ref{prop:limit-stratified-propsnsoty-score-coarse}, the efficiency gain should stay positive even for large $n.$
Note that the number of cells in partition (1) is $K = 2^{2 n}$ while that of partition (2) is $K = 1$ if $n = 0$ and $K = 2^{2 n - 1} + 1$ if $n \geq 1.$
From Theorem \ref{thm:speb-stratified} and the formula (\ref{eq:efficiency-bound-ATT}) of the efficiency bound of ATT, for a given partition $\bigsqcup_{k = 1}^K \cX_k,$ we have $V_{\mathrm{ATT}}^k = 8 / 3,$ $V_{\mathrm{ATT}}^{uk} = 13 / 3,$ and
\begin{align*}
    V_{\mathrm{ATT}}^p
    =
    \frac{8}{3}
    +
    \sum_{k = 1}^K
    \left(
        \bE [2 X_1 + 4 X_2 \mid X \in \cX_k]
    \right)^2
    \bP \left(
        X \in \cX_k
    \right)
    .
\end{align*}

\begin{figure}
    \centering
    \includegraphics[width=\linewidth]{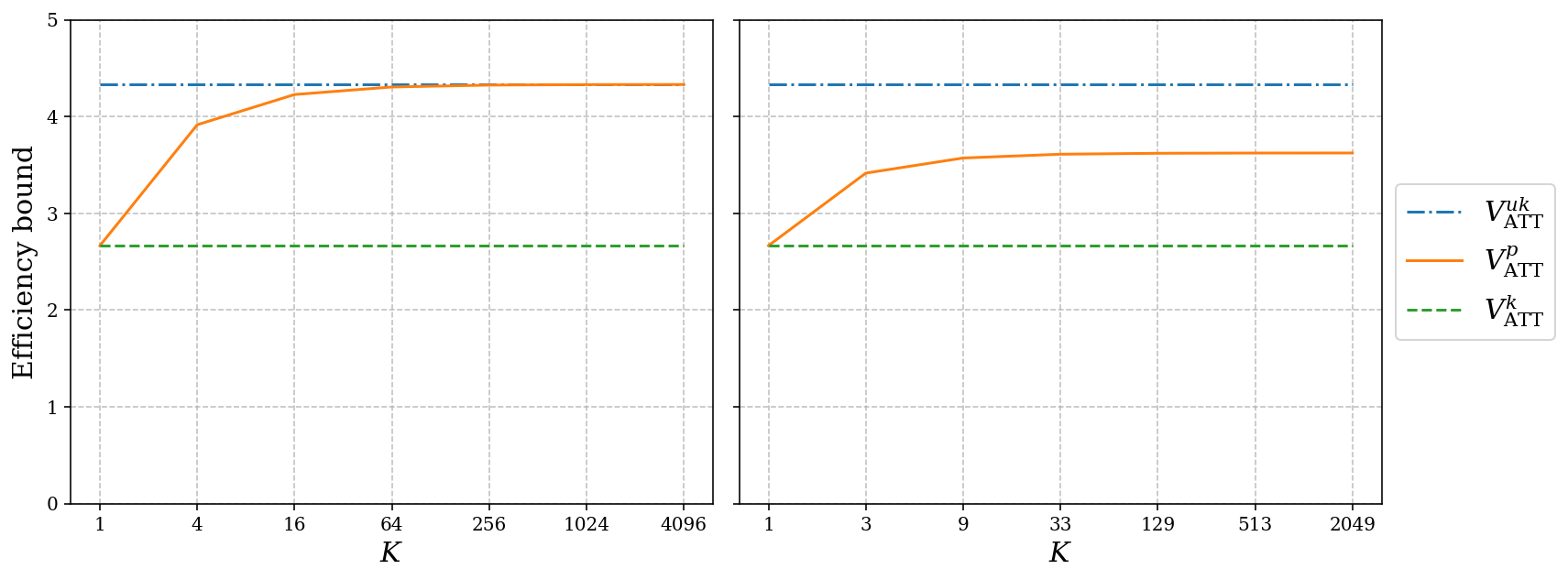}
    \caption{The efficiency bounds. The horizontal axis is the number $K = |\cI_n^i|$ of cells for $i = 1, 2.$ The left panel corresponds to partition (1), and the right panel corresponds to partition (2). The efficiency bounds $V_{\mathrm{ATT}}^{uk}$ (blue, dash-dotted), $V_{\mathrm{ATT}}^p$ (orange, solid), and $V_{\mathrm{ATT}}^k$ (green, dashed) are shown.}
    \label{fig:efficiency-gains}
\end{figure}

Figure \ref{fig:efficiency-gains} reports the three efficiency bounds for different partitions.
As the theory predicts, $V_{\mathrm{ATT}}^{uk} - V_{\mathrm{ATT}}^p$ approaches zero as $K$ becomes large in the left panel while it remains positive in the right panel.
In the left panel, we see that the value of knowing the partition drops drastically, as $K$ increase from $1$ to $4,$ which implies that even relatively coarse partitions are almost useless in terms of efficiency.
In the right panel, the efficiency gain $V_{\mathrm{ATT}}^{uk} - V_{\mathrm{ATT}}^p \approx 0.7$ for large $K$ can be interpreted as the value of knowing that the propensity score is constant on the right half of $\cX.$

\printbibliography

@article{hahn1998role,
  title={On the role of the propensity score in efficient semiparametric estimation of average treatment effects},
  author={Hahn, Jinyong},
  journal={Econometrica},
  pages={315--331},
  year={1998},
  publisher={JSTOR}
}

@article{frolich2004note,
  title={A note on the role of the propensity score for estimating average treatment effects},
  author={Fr{\"o}lich, Markus},
  journal={Econometric Reviews},
  volume={23},
  number={2},
  pages={167--174},
  year={2004},
  publisher={Taylor \& Francis}
}

@book{bickel1993efficient,
  title={Efficient and adaptive estimation for semiparametric models},
    author={Bickel, PJ and Klaassen, CAJ and Ritov, Y and Wellner, JA},
  volume={4},
  year={1993},
  publisher={Springer}
}

@article{hirano2003efficient,
  title={Efficient estimation of average treatment effects using the estimated propensity score},
  author={Hirano, Keisuke and Imbens, Guido W and Ridder, Geert},
  journal={Econometrica},
  volume={71},
  number={4},
  pages={1161--1189},
  year={2003},
  publisher={Wiley Online Library}
}

@article{chen2008semiparametric,
  title={Semiparametric efficiency in GMM models with auxiliary data},
  author={Chen, Xiaohong and Hong, Han and Tarozzi, Alessandro},
  journal={The Annals of Statistics},
  volume={36},
  number={2},
  pages={808--843},
  year={2008},
  publisher={Institute of Mathematical Statistics}
}

@article{lee2018efficient,
  title={Efficient propensity score regression estimators of multivalued treatment effects for the treated},
  author={Lee, Ying-Ying},
  journal={Journal of Econometrics},
  volume={204},
  number={2},
  pages={207--222},
  year={2018},
  publisher={Elsevier}
}

@article{cattaneo2010efficient,
  title={Efficient semiparametric estimation of multi-valued treatment effects under ignorability},
  author={Cattaneo, Matias D},
  journal={Journal of Econometrics},
  volume={155},
  number={2},
  pages={138--154},
  year={2010},
  publisher={Elsevier}
}

@article{farrell2015robust,
  title={Robust inference on average treatment effects with possibly more covariates than observations},
  author={Farrell, Max H},
  journal={Journal of Econometrics},
  volume={189},
  number={1},
  pages={1--23},
  year={2015},
  publisher={Elsevier}
}

@article{chernozhukov2018double,
  title={Double/debiased machine learning for treatment and structural parameters.},
  author={Chernozhukov, Victor and Chetverikov, Denis and Demirer, Mert and Duflo, Esther and Hansen, Christian and Newey, Whitney and Robins, James},
  journal={Econometrics Journal},
  volume={21},
  number={1},
  year={2018}
}

@article{firpo2007efficient,
  title={Efficient semiparametric estimation of quantile treatment effects},
  author={Firpo, Sergio},
  journal={Econometrica},
  volume={75},
  number={1},
  pages={259--276},
  year={2007},
  publisher={Wiley Online Library}
}

@article{chernozhukov2022locally,
  title={Locally robust semiparametric estimation},
  author={Chernozhukov, Victor and Escanciano, Juan Carlos and Ichimura, Hidehiko and Newey, Whitney K and Robins, James M},
  journal={Econometrica},
  volume={90},
  number={4},
  pages={1501--1535},
  year={2022},
  publisher={Wiley Online Library}
}

@article{imbens2000role,
  title={The role of the propensity score in estimating dose-response functions},
  author={Imbens, Guido W},
  journal={Biometrika},
  volume={87},
  number={3},
  pages={706--710},
  year={2000},
  publisher={Oxford University Press}
}

@article{rosenbaum1983central,
  title={The central role of the propensity score in observational studies for causal effects},
  author={Rosenbaum, Paul R and Rubin, Donald B},
  journal={Biometrika},
  volume={70},
  number={1},
  pages={41--55},
  year={1983},
  publisher={Oxford University Press}
}

@article{rosenbaum1984reducing,
  title={Reducing bias in observational studies using subclassification on the propensity score},
  author={Rosenbaum, Paul R and Rubin, Donald B},
  journal={Journal of the American statistical Association},
  volume={79},
  number={387},
  pages={516--524},
  year={1984},
  publisher={Taylor \& Francis}
}

@article{newey1990semiparametric,
  title={Semiparametric efficiency bounds},
  author={Newey, Whitney K},
  journal={Journal of applied econometrics},
  volume={5},
  number={2},
  pages={99--135},
  year={1990},
  publisher={Wiley Online Library}
}

@article{yang2016propensity,
  title={Propensity score matching and subclassification in observational studies with multi-level treatments},
  author={Yang, Shu and Imbens, Guido W and Cui, Zhanglin and Faries, Douglas E and Kadziola, Zbigniew},
  journal={Biometrics},
  volume={72},
  number={4},
  pages={1055--1065},
  year={2016},
  publisher={Wiley Online Library}
}

@article{imai2004causal,
  title={Causal inference with general treatment regimes: Generalizing the propensity score},
  author={Imai, Kosuke and Van Dyk, David A},
  journal={Journal of the American Statistical Association},
  volume={99},
  number={467},
  pages={854--866},
  year={2004},
  publisher={Taylor \& Francis}
}

@article{chong2016iron,
  title={Iron deficiency and schooling attainment in Peru},
  author={Chong, Alberto and Cohen, Isabelle and Field, Erica and Nakasone, Eduardo and Torero, Maximo},
  journal={American Economic Journal: Applied Economics},
  volume={8},
  number={4},
  pages={222--255},
  year={2016},
  publisher={American Economic Association 2014 Broadway, Suite 305, Nashville, TN 37203-2425}
}

@article{bugni2019inference,
  title={Inference under covariate-adaptive randomization with multiple treatments},
  author={Bugni, Federico A and Canay, Ivan A and Shaikh, Azeem M},
  journal={Quantitative Economics},
  volume={10},
  number={4},
  pages={1747--1785},
  year={2019},
  publisher={Wiley Online Library}
}

@article{hong2020inference,
  title={Inference on finite-population treatment effects under limited overlap},
  author={Hong, Han and Leung, Michael P and Li, Jessie},
  journal={The Econometrics Journal},
  volume={23},
  number={1},
  pages={32--47},
  year={2020},
  publisher={Oxford University Press}
}

@article{herren2023true,
  title={On true versus estimated propensity scores for treatment effect estimation with discrete controls},
  author={Herren, Andrew and Hahn, P Richard},
  journal={arXiv preprint arXiv:2305.11163},
  year={2023}
}

@article{chen2007large,
  title={Large sample sieve estimation of semi-nonparametric models},
  author={Chen, Xiaohong},
  journal={Handbook of econometrics},
  volume={6},
  pages={5549--5632},
  year={2007},
  publisher={Elsevier}
}

@article{bai2024efficiency,
  title={On the efficiency of finely stratified experiments},
  author={Bai, Yuehao and Liu, Jizhou and Shaikh, Azeem M and Tabord-Meehan, Max},
  journal={arXiv preprint arXiv:2307.15181},
  year={2024}
}

@article{cytrynbaum2021designing,
  title={Designing representative and balanced experiments by local randomization},
  author={Cytrynbaum, Max},
  journal={arXiv preprint arXiv:2111.08157},
  year={2021},
  publisher={Technical report}
}

@article{tabord2023stratification,
  title={Stratification trees for adaptive randomisation in randomised controlled trials},
  author={Tabord-Meehan, Max},
  journal={Review of Economic Studies},
  volume={90},
  number={5},
  pages={2646--2673},
  year={2023},
  publisher={Oxford University Press US}
}

\end{document}